\documentclass[letterpaper, 10 pt, conference]{IEEEtran}  

\IEEEoverridecommandlockouts                              

\usepackage{graphicx}
\usepackage{epsfig} 
\usepackage{epstopdf}
\usepackage{subfigure}
\usepackage{amsmath} 
\usepackage{amssymb}  
\usepackage[T1]{fontenc}
\usepackage{mathtools}
\usepackage{latexsym}

\usepackage{amsfonts}
\usepackage[german,english]{babel}
\usepackage{color}
\usepackage{todonotes}
\usepackage{soul}
\usepackage{multirow}

\usepackage{wrapfig}

\graphicspath{{./images/}}

\newenvironment{proof}[1]{\vspace{.1cm}\noindent{\sc Proof #1. }\hspace{0.05cm}\,\,}{$\hfill\Box$\vspace{.1cm}} 
\newtheorem{theorem}            {Theorem}[section] 
\newtheorem{definition}         [theorem]{Definition}

\newtheorem{lemma}              [theorem]{Lemma} 
\newtheorem{proposition}		[theorem]{Proposition} 
\newtheorem{corollary}		[theorem]{Corollary}
 
\newtheorem{example}		[theorem]{Example}

\newcommand{\dd}{{\rm d}\hbox{\hskip 0.5pt}}

\newcommand{\R}{\mathbb{R}}

\newcommand\disp{\displaystyle}

\newcommand\iintdl{\iint\displaylimits}

\title{\LARGE \bf
On the characterization of butterfly and multi-loop hysteresis behavior
}

\author{M.A. Vasquez-Beltran, B. Jayawardhana, R. Peletier%
    \thanks{*This paper is based on research developed in the DSSC Doctoral Training Programme, co-funded through a Marie Skłodowska-Curie COFUND (DSSC 754315).}
    \thanks{$^{1}$M.A. Vasquez Beltran and B. Jayawardhana are with the Engineering and Technology Institute Groningen, Faculty of Science and Engineering, University of Groningen, 9747AG Groningen, The Netherlands {\tt\small \{m.a.vasquez.beltran;b.jayawardhana\}@rug.nl}}%
    \thanks{$^{2}$R. Peletier is with the Kapteyn Astronomical Institute, Faculty of Science and Engineering, University of Groningen, 9747AG Groningen, The Netherlands {\tt\small r.peletier@rug.nl}}%
}

\begin{document}
    
\maketitle
\thispagestyle{empty}
\pagestyle{empty}


\begin{abstract}
While it is widely used to represent hysteresis phenomena with unidirectional-oriented loops, we study in this paper the use of Preisach operator for describing hysteresis behavior with multidirectional-oriented loops. This complex hysteresis behavior is commonly found in advanced materials, such as, shape-memory alloys or piezoelectric materials, that are used for high-precision sensor and actuator systems. We provide characterization of the Preisach operators exhibiting such input-output behaviors and we show the richness of the operators that are capable of producing intricate loops.
\end{abstract}


\section{Introduction}

The term “hysteresis” comes from the Greek “to lag behind” and was originally coined by Ewing in 1885 to describe a phenomenon occurring in the magnetization process of soft iron caused by reversal and cyclic changes of the input magnetic field. 
Currently, hysteresis is known to be present in several classes of physical systems such as ferroelectric and ferromagnetic materials, shape memory alloys, and mechanical systems with friction. Hysteresis represents a quasi-static dependence between the input and output of a system whose phase plot describes particular curves known as hysteresis loops \cite{Bernstein2007}.

Hysteresis is a non-linear phenomenon which has been represented in various different mathematical formulation. 
One major distinction in representing hysteresis is between 
the physics-based models and phenomenological models. The former focuses on describing the hysteresis phenomenon from the particular physical relations of the system under consideration, whereas the latter focuses on the empirical description of the input-output behavior. 
Due to the simplicity and ability to encapsulate many typical characteristics of hysteresis behavior, 
the phenomenological models have been widely studied for the past decades. Two of the phenomenological models that have been widely used are the Preisach operator \cite{Preisach1935, Mayergoyz1988}, whose formulation incorporates other operator-based models such as the Prandlt operator; and the Duhem model \cite{Visintin1994} whose formulation incorporates other models based on non-smooth integro-differential equations such as the Bouc-Wen model and the Dahl model.

In literature, there are numerous works that investigate the mathematical properties of these phenomenological models \cite{Macki1993, Gu2016, Brokate1996}. Subsequently, these characterization works are directly applicable for  
the stability analysis of and the control design for systems that consist of sub-systems exhibiting hysteresis behavior. For instance,
the construction of an approximate inverse model is pursued in \cite{Iyer2005} in order to stabilize a control system containing hysteretic elements. 
In recent years, it has been shown that these phenomenological hysteresis models can exhibit passivity/dissipativity property which is a typical property of physical systems. The dissipativity of Duhem model  
has been shown in \cite{Jayawardhana2012,Ouyang2013} while that of Preisach model is presented in \cite{Gorbet2001}. These dissipativity properties are closely related to the orientation of the hysteresis loops.

Despite these numerous endeavors, most of the works in hysteresis modeling have focused on the characterization of the hysteresis behavior whose phase plot describes single-oriented loop as illustrated in Fig. \ref{fig:loop_single}. Common examples of single-oriented loop occurs in the relation between polarization and electric fields of piezoelectric materials or the relation between magnetization and magnetic field of magnetostrictive materials. However, there exists another class of hysteresis behavior reported in literature (see, for instance, \cite{Waldmann2002,Sahota2004,Linnemann2009,Linhart2011}) whose phase plot describes two loops with opposite orientations and connected at an intersection point as depicted in Fig. \ref{fig:loop_butterfly}. From the resemblance to the wings of a butterfly, this behavior is known as butterfly hysteresis behavior. Examples of this type of hysteresis behavior occur in the relation between strain and electric field of piezoelectric materials and the relation between strain and magnetic field of magnetostrictive materials.

To the authors' best knowledge, there are two works providing mathematical analysis for the modeling of butterfly hysteresis behavior. Firstly, in \cite{Pozo2009} a modified Bouc–Wen is studied which can describe a particular class of asymmetric double hysteresis loop behavior by introducing position and/or acceleration information into the model equation. Moreover, when the parameters of the model satisfy particular conditions, it has been shown that passivity property holds \cite{Pozo2015}. Secondly, in \cite{Drincic2011}, a framework to transform butterfly loops to single-oriented loops is proposed. Although this approach facilitates the use of the well-studied hysteresis models and enables the possibility of implementing some of its known control strategies in systems including elements that exhibit hysteresis with butterfly loops, it relies on the existence of a convex mapping and restricts the loop shape to have exactly two minima with the same value.

In this work, we extend the results and include the proof of propositions in \cite{Jayawardhana2018} where we introduced a Preisach hysteresis operator capable of exhibiting butterfly loops. It is used to 
model the relation between strain and electric field of a particular piezoelectric material. 
We firstly present the analysis of a class of Preisach operators whose weighting function has one positive and one negative domain. We show that under mild assumptions over the distributions of these domains, the input-output behavior of Preisach operator can exhibits butterfly loops. Subsequently, we introduce a general class of Preisach operator whose weighting function can assume 
more than one positive and one negative domain. We show that the input-output behavior of these operators can exhibit hysteresis loops with two or more sub-loops. Finally, we present the stability analysis of a Lur'e system with multi-loop hysteresis element based on the results in \cite{Vasquez-Beltran2020}.

\begin{figure}
    \centering
    \includegraphics[width=0.4\linewidth]{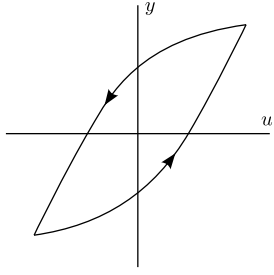}
    \caption{An illustration of a simple hysteresis loop $\mathcal H_{u,y}$ that typically describes the relation between polarization and electric fields of piezoelectric materials or the relation between magnetization and magnetic field of magnetostrictive materials. \label{fig:loop_single}}
\end{figure}

\begin{figure}
    \centering
    \includegraphics[width=0.4\linewidth]{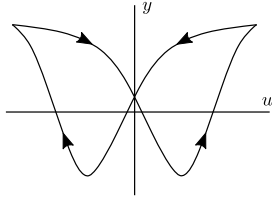}
    \caption{An illustration of a butterfly loop $\mathcal H_{u,y}$ that can describe the relation between strain and electric field of piezoelectric materials or the relation between strain and magnetic field in magnetostrictive materials. \label{fig:loop_butterfly}}
\end{figure}

This paper is organized as follows. In Section \ref{sec:preliminaries} we give some preliminaries that include the definition of hysteresis operator, operators with clockwise and counterclockwise input-output behavior, and the standard definition of the Preisach operator. In Section \ref{sec:preisach_butterfly_operator}, we present the Preisach butterfly operator and the proof of the main results in \cite{Jayawardhana2018}. In Section \ref{sec:preisach_multiloop_operator} we introduce a characterization of the self-intersections in a hysteresis loop and their relation to the weighting function of a Preisach multi-loop operator and Section \ref{sec:lure_multiloop} presents the absolute stability analysis of the Lur'e system with a Preisach multi-loop operator. Finally, in Section \ref{sec:conclusions} we give the conclusions.


\section{Preliminaries}\label{sec:preliminaries}

{\bf Notation.} We denote the spaces of piecewise continuous, absolute continuous and continuous differentiable functions $f:U\to Y$ by $C_{\text{pw}}(U,Y)$, $AC(U,Y)$ and $C(U,Y)$, respectively. 
A function $f:U\to Y$ is {\em monotonically increasing (resp. decreasing)} if for every $u_1,u_2\in U$ such that $u_1<u_2$ we have that $f(u_1)\leq f(u_2)$ (resp. $f(u_1)\geq f(u_2)$).

\subsection{Clockwise, counterclockwise and butterfly loops}

In order to define hysteresis operators (following the formulation in \cite{Logemann2003}), 
we introduce below three auxiliary concepts: {\em time-transformation}, {\em rate-independent operator} and {\em causal operator}.
\begin{definition}\label{def:time_transformation}
    A function $\phi:\R_+\to\R_+$ is called a {\em time transformation} if $\phi(t)$ is continuous and increasing with $\phi(0)=0$ and $\lim_{t\to \infty}\phi(t)=\infty$. $\hfill \triangle$
\end{definition}
\begin{definition}\label{def:rate_independent}
    An operator $\Phi$ is said to be {\em rate independent} if 
    \begin{equation*}
        \big(\Phi(u\circ \phi)\big)(t) = \Phi(u)\circ \phi(t)
    \end{equation*}
    holds for all $u\in AC(\R_+,\R)$, $t\in\R_+$ and all admissible time transformation $\phi$. $\hfill \triangle$
\end{definition}
\begin{definition}\label{def:causal}
    The operator $\Phi$ is said to be {\em causal} if for all $\tau>0$ and all $u_1,\,u_2\in AC(\R_+,\R)$ it holds that
    \begin{multline*}
        u_1(t)=u_2(t) \ \ \ \forall t\in [0,\tau] \\ 
        \Rightarrow \big(\Phi(u_1)\big)(t) = \big(\Phi(u_2)\big)(t) \ \ \ \forall t\in[0,\tau].
    \end{multline*} $\hfill \triangle$
\end{definition}\vspace{0.1cm}
Based on the definitions above, a hysteresis operator is defined formally as follows.

\begin{definition}\label{def:hysteresis_operator}
    An operator $\Phi$ is called a {\em hysteresis operator} if $\Phi$ is causal and rate-independent. $\hfill \triangle$
\end{definition}\vspace{0.1cm}

For the past decades, hysteresis operators have been widely studied and characterized (see, for instance, the exposition in  
\cite{Brokate1996,Mayergoyz2003,Visintin1994}). 
When a periodic input is applied to a hysteresis operator, the input-output phase plot will undergo a periodic closed orbit which is commonly referred to as  hysteresis loop. Similar to the periodic input-output map introduced in \cite[Definition 2.2]{JinHyoungOh2005} for Duhem models, we can define a {\em hysteresis loop} as follows.\\

\begin{definition}\label{def:hysteresis_loop}
    Consider a hysteresis operator $\Phi$ and an input-output pair $(u,y)$ with $y=\Phi(u)$. Let $u$ be periodic with a period of $T>0$, with one maximum $u_{\max}\in\R$ and with one minimum $u_{\min}\in\R$ in its periodic interval. Assume that there exists a constant $t_p\geq0$ such that $y$ is periodic in the interval $[t_p,\infty)$. 
    The periodic orbit given by $\mathcal{H}_{u,y} = \left\{ (u(t),y(t))\ |\ t\in[t_p,\infty) \right\}$ is called a {\em hysteresis loop} if there exists a $\upsilon \in \R$ such that
    \[
    \text{card}\big(\{(\upsilon,\gamma)\in \mathcal{H}_{u,y}\ |\ \gamma\in\R\}\big)=2,
    \]
    where card denotes the cardinality of a set. $\hfill \triangle$
\end{definition}\vspace{0.1cm}

In other words, the hysteresis loop as defined above means that the curve defined by $\mathcal H_{u,y}$ can have at most two  elements $(\upsilon,\gamma_1)$ and $(\upsilon,\gamma_2)$ for any admissible point $\upsilon$. This definition admits also multiple input-output loops that will be studied further in this paper.

As shown in \cite{Ouyang2013,Jayawardhana2012,Valadkhan2009,Padthe2005}, the input-output behavior of hysteresis operators can be classified by the type of hysteresis loops they produce. Simple hysteresis loops
can have a clockwise or counterclockwise orientation which is given in terms of the signed-area enclosed by its phase plot.
Following from Green's theorem, the signed-area enclosed by an input-output pair $(u,y)$ that forms a closed curve in an interval $[t_1, t_2]$ is given by
\begin{equation}\label{eq:closed_area}
    \mathbb{A}:=\frac{1}{2}\int_{t_1}^{t_2}\big(u(\tau)\dot y(\tau) -y(\tau)\dot u(\tau) \big) \dd \tau.
\end{equation}
Hence, generalizing this notion we can say that a hysteresis loop $\mathcal{H}_{u,y}$ is clockwise (resp. counterclockwise) if its signed-area $\mathbb{A}$ given by \eqref{eq:closed_area} with $t_1\geq t_p$ and $t_2=t_1+T$  satisfies $\mathbb{A}<0$ (resp. $\mathbb{A}>0$). In a similar manner, we say that a hysteresis operator $\Phi$ exhibits clockwise (resp. counterclockwise) input-output behavior if there exists at least one hysteresis loop $\mathcal{H}_{u,y}$ corresponding to an input-output pair $(u,y)$ with $y=\Phi(u)$ which is clockwise (resp. counterclockwise).

Based on these concepts, we can study  hysteresis operators $\Phi$ that give rise to butterfly loops using the enclosed signed-area as in \eqref{eq:closed_area} of the resulting hysteresis loops. However, as depicted in Fig. \ref{fig:loop_butterfly}, we would like to note that the so-called butterfly loops are composed of two sub-loops connected by a self intersection point where one of loops is clockwise and the other is  counterclockwise. Thus, the total signed-area of the butterfly loop could be either positive or negative depending on the difference between the individual signed-area of each sub-loop. For this reason, we define a {\em butterfly hysteresis operator} as follows.
\begin{definition}\label{def:butterfly_operator}
    A hysteresis operator $\Phi$ is called a {\em butterfly hysteresis operator} if there exists a hysteresis loop $\mathcal{H}_{u,y}$ with $y=\Phi(u)$ 
    such that $\mathbb{A} = 0$, where $\mathbb{A}$ is defined as in \eqref{eq:closed_area} with $t_1\geq t_p$ and $t_2=t_1+T$. $\hfill \triangle$
\end{definition}\vspace{0.2cm}
We remark that this definition of {\em butterfly hysteresis operator} is equivalent to the one that we have introduced in \cite[Definition 2.2]{Jayawardhana2018}. In the present paper we have used the concept of {\em hysteresis loop} in Definition \ref{def:hysteresis_loop} for defining the butterfly hysteresis operator with a particular input-output pair $(u,y)$ that forms a periodic orbit, e.g. on the interval $[t_1,t_1+T]$.


\subsection{The clockwise and counterclockwise relay operator}

One of the simplest hysteresis operators is the relay operator which we introduce below in its counterclockwise and clockwise versions. We define the counterclockwise relay operator $\mathcal R^{\circlearrowleft}_{\alpha,\beta}: AC(\R_+,\R)\times \{-1, 1\} \to C_{\text{pw}}(\R_+,\R)$ with switching parameters $\alpha>\beta$ and initial condition $r_0$ by 
\begin{equation}\label{eq:relay_operator_ccw}
    \big(\mathcal R^{\circlearrowleft}_{\alpha,\beta}(u,r_0)\big)(t) :=
    \left\{ \begin{array}{ll}
        +1 & \text{if } u(t)>\alpha, \\
        -1 & \text{if } u(t)<\beta, \\
        \big(\mathcal R^{\circlearrowleft}_{\alpha,\beta}(u,r_0)\big)(t_-) & 
        \begin{aligned}
            \text{if } &\beta\leq u(t) \leq\alpha, \\ 
            &\quad\text{and } t>0,
        \end{aligned} \\
        r_0 & \begin{aligned}
            \text{if } &\beta\leq u(t) \leq\alpha, \\
            &\quad\text{and } t=0.
        \end{aligned}
    \end{array}\right.
\end{equation}

\begin{figure}
    \centering
    \includegraphics[width=0.6\linewidth]{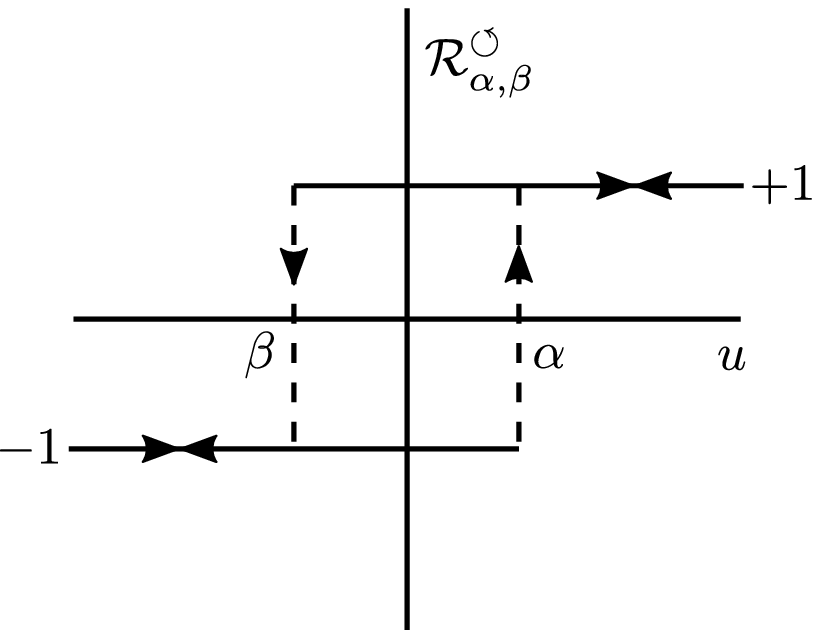}
    \caption{Input-output phase plot of the counterclockwise relay operator $\mathcal{R}^\circlearrowleft_{\alpha,\beta}$ as defined in \eqref{eq:relay_operator_ccw} \label{fig:ccw_relay}}
\end{figure}

Similarly, we define the clockwise relay operator $\mathcal R^{\circlearrowright}_{\alpha,\beta}: AC(\R_+,\R)\times \{-1, 1\} \to C_{\text{pw}}(\R_+,\R)$ with switching parameters $\alpha>\beta$ and initial condition $r_0$ by
\begin{equation}\label{eq:relay_operator_cw}
    \big(\mathcal R^{\circlearrowright}_{\alpha,\beta}(u,r_0)\big)(t) :=
    \left\{ \begin{array}{ll}
        -1 & \text{if } u(t)>\alpha, \\
        +1 & \text{if } u(t)<\beta, \\
        \big(\mathcal R^{\circlearrowright}_{\alpha,\beta}(u,r_0)\big)(t_-) & 
        \begin{aligned}
            \text{if } &\beta\leq u(t) \leq\alpha, \\ 
            &\quad\text{and } t>0,
        \end{aligned} \\
        -r_0 & \begin{aligned}
            \text{if } &\beta\leq u(t) \leq\alpha, \\
            &\quad\text{and } t=0.
        \end{aligned}
    \end{array}\right.
\end{equation}
We note that for a specified initial condition $r_0\in\{-1,+1\}$ both relay operators are hysteresis operators in the sense of Definition \ref{def:hysteresis_operator} with the form
\begin{equation*}
    \begin{aligned}
        \Phi(u) = \mathcal{R}_{\alpha,\beta}^{\circlearrowleft}(u,\,r_0) 
        \qquad\text{and}\qquad
        \Phi(u) = \mathcal{R}_{\alpha,\beta}^{\circlearrowright}(u,\,r_0).
    \end{aligned}
\end{equation*}
The input-output phase plot of both relay operators is illustrated in Figs. \ref{fig:ccw_relay} and \ref{fig:cw_relay}. It can be observed from definitions \eqref{eq:relay_operator_ccw} and \eqref{eq:relay_operator_cw} that for equal initial condition $r_0$, the output of a counterclockwise relay is equivalent to the negative output of a clockwise operator for every input, i.e. $\mathcal R^{\circlearrowleft}_{\alpha,\beta} = - \mathcal R^{\circlearrowright}_{\alpha,\beta}$.

\begin{figure}
    \centering
    \includegraphics[width=0.6\linewidth]{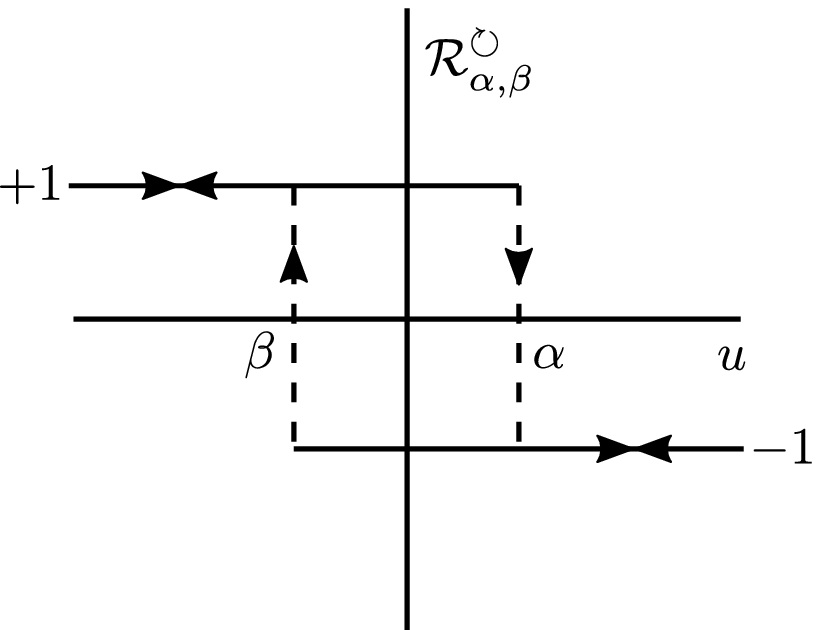}
    \caption{Input-output phase plot of the clockwise relay operator $\mathcal{R}^\circlearrowright_{\alpha,\beta}$ as defined in \eqref{eq:relay_operator_cw} \label{fig:cw_relay}}
\end{figure}

\subsection{Preisach Operator}

The Preisach operator is, roughly speaking, the weighted integral of all infinitesimal (counterclockwise) relay operators, also known as hysterons, whose switching values satisfy $\alpha>\beta$. To provide a formal definition of this, we need to introduce two concepts. Firstly, we denote by $P$ the admissible plane of relay operators defined by $P:=\{(\alpha,\beta)\in\R^2\ |\ \alpha>\beta\}$ and it is commonly referred to as the Preisach plane. 
Secondly, we denote $\mathcal I$ the set of interfaces where each interface $L\in\mathcal{I}$ is a monotonically decreasing staircase curve parameterized in the form $L = \{ \sigma(\gamma)\in P\ |\ \gamma\in\R_+ \}$ by a function $\sigma(\gamma)\in C(\R_+,P)$ which satisfies $\lim_{\gamma\to \infty} \lVert\sigma(\gamma)\rVert=\infty$ and $\sigma(0)=(\alpha,\alpha)$ for some $\alpha \in \R$. By monotonically decreasing $L$ we mean that for every pair $(\alpha_1,\beta_1),(\alpha_2,\beta_2)\in L$ we have that $\beta_1\leq\beta_2$ implies $\alpha_1\geq\alpha_2$. Based on these concepts, the Preisach operator $\mathcal{P}:AC(\R_+,\R)\times \mathcal{I}\to AC(\R_+,\R)$ is formally expressed by 
\begin{equation}\label{eq:preisach_operator}
    \begin{aligned}
        &\big( \mathcal{P}(u,L_0) \big)(t) := \\
        &\quad\iintdl_{(\alpha,\beta)\in P}
        \mu(\alpha,\beta)\ 
        \big( 
        \mathcal R^{\circlearrowleft}_{\alpha,\beta} 
        ( u, r_{\alpha,\beta}(L_0) )
        \big) (t) \
        \dd{\alpha} \dd\beta
    \end{aligned}
\end{equation}
where $\mu\in C_{\text{pw}}(P,\R)$ is a weighting function, $L_0 \in \mathcal{I}$ is the initial interface and $r_{\alpha,\beta}:\mathcal{I}\to\{-1,+1\}$ is an auxiliary function that determines the initial condition of every relay $\mathcal R^{\circlearrowleft}_{\alpha,\beta}$ according to its position $(\alpha,\beta)\in P$ with respects to the initial interface $L_0$ and is defined by
\begin{equation*}
    \begin{aligned}
        &r_{\alpha,\beta}(L_0) := \\
        &\quad\left\{ \begin{array}{ll}
            +1 & \begin{aligned}    
                \text{if }L_0 \cap \left\{(\alpha_1,\beta_1)\in P\ |\ \alpha\leq\alpha_1,\, \beta\leq\beta_1\right\} \neq\emptyset,
            \end{aligned} \\
            -1 & \text{otherwise}.
        \end{array}
        \right.
    \end{aligned}
\end{equation*}
In other words, the value of the function $r_{\alpha,\beta}$ will be $+1$ if the point $(\alpha,\beta)\in P$ is above the initial interface $L_0$, and will be $-1$ if the point $(\alpha,\beta)\in P$ is below the initial interface $L_0$.
It is important to note from \eqref{eq:relay_operator_ccw} that the actual initial state of a relay operator $\mathcal R^{\circlearrowleft}_{\alpha,\beta}$ is determined by $r_0$ only when $\beta\leq u(0)\leq\alpha$. This can produce an inconsistency between the values of the function $r_{\alpha,\beta}(L_0)$ and the actual initial state of relay operators $\big( \mathcal R^{\circlearrowleft}_{\alpha,\beta} ( u,r_{\alpha,\beta}(L_0) ) \big) (0)$ with $\alpha<u(0)$ or $\beta>u(0)$. Therefore, for well-posedness, we assume in general that the initial interface $L_0$ in the Preisach operator \eqref{eq:preisach_operator} satisfies  $\big(u(0),u(0)\big)\in L_0$.
As with the relay operator, the Preisach operator is a hysteresis operator in the sense of Definition \ref{def:hysteresis_operator} for specified initial conditions $L_0\in\mathcal{I}$ in the form $\Phi(u) = \mathcal{P}(u,L_0)$.

As one of the most important hysteresis operators, the dynamic behavior and geometric interpretation of the Preisach operator defined by \eqref{eq:preisach_operator} has been studied well in literature. 
Fundamentally, it can be said that the output of the Preisach operator is determined instantaneously with the variations of the input as all the relays in $P$ react instantaneously and simultaneously to the applied input $u$. For this reason, the initial interface $L_0$ evolves continuously and at every time instance $t\geq 0$ there exists an interface $L_t \in \mathcal{I}$ that divides the Preisach plane into two subdomains $P^+_t$ and $P^-_t$ where all relays $\mathcal{R}_{\alpha,\beta}$ with $({\alpha,\beta})\in P^+_t$ are in state $+1$ while all relays $\mathcal{R}_{\alpha,\beta}$ with $({\alpha,\beta})\in P^-_t$ are in state $-1$.

\section{Preisach Butterfly Operator}\label{sec:preisach_butterfly_operator}

The classical definition of the Preisach operator \cite{Preisach1935, Mayergoyz1988} considers that the weighting function $\mu$ is positive with counterclockwise relays as in \eqref{eq:relay_operator_ccw}. This assumption restricts all realizable hysteresis loops to be counterclockwise (for any periodic input-output pairs). On the other hands, when negative weighting function $\mu$ is assumed (with the same counterclockwise relays) then 
the realizable hysteresis loops become clockwise (see \cite{Visone2015}). This is mainly due to the fact that 
$\mathcal R^{\circlearrowleft}_{\alpha,\beta} = - \mathcal R^{\circlearrowright}_{\alpha,\beta}$, in which case, a Preisach operator with negative weighting function $\mu$ and counterclockwise relays is equivalent to a Preisach operator with a positive weighting function $\mu$ and clockwise relays. 

Inspired by the latter observation and considering that a {\em butterfly hysteresis operator} must exhibit both clockwise and counterclockwise input-output behavior, it is intuitive to think that a Preisach operator with both clockwise and counterclockwise relays and a sign-definite weighting function, or equivalently with only clockwise or only counterclockwise relays and whose weighting function $\mu$ has positive and negative domains, can, under certain conditions, exhibit a {\em butterfly hysteresis operator}. 
We formalize this idea introducing a Preisach operator whose weighting function $\mu$ has particular structure and whose input-output behavior exhibits a butterfly loop that satisfies the zero signed-area condition of Definition \ref{def:butterfly_operator}. For this purpose, let us introduce the following lemma that allows us to compute the enclosed area of the hysteresis loop of a single relay operator.
\begin{lemma}\label{lemma:relay_area}
    Consider the counterclockwise relay operator $\mathcal R^{\circlearrowleft}_{\alpha,\beta}$ as in \eqref{eq:relay_operator_ccw} with $\alpha>\beta$. For every periodic signal $u$ with a period of $T$ and with one maximum $u_{\max}\in\R$ and one minimum $u_{\min}\in\R$ in the periodic interval, the signed-area $\mathbb A$ corresponding to the input-output pair $(u,y)$ with $y=\mathcal R^{\circlearrowright}_{\alpha,\beta}(u,r_0)$ and $r_0\in\{-1,+1\}$ is given by
    \vspace*{-0.2cm}
    \begin{equation*}
        \mathbb A = \left\{\begin{array}{ll} 2(\alpha-\beta) & \text{if }u_{\min}<\beta \text{ and }u_{\max}>\alpha \\
        0 & \text{otherwise.}\end{array}\right.
    \end{equation*}
\end{lemma}

\begin{proof}{}
    It can be checked that the influence of relay initial condition to the output signal $y=\mathcal R^{\circlearrowright}_{\alpha,\beta}(u,r_0)$ disappears after one period. In other words, the pair $(u,y)$ will form a 
    hysteresis loop $\mathcal{H}_{u,y}$
    or a line in the time interval $[t_1,\infty)$ with $t_1\geq t_p=T$. When the range of the input covers both switching points of the relay, i.e., $u_{\min} <\beta$ and $u_{\max} > \alpha$, 
    then the relay will switch periodically forming a
    hysteresis loop $\mathcal{H}_{u,y}$ 
    as in Fig. \ref{fig:ccw_relay}. Otherwise, it will be a line. If it is a 
    hysteresis loop $\mathcal{H}_{u,y}$ 
    then the signed-area that is produced 
    is given by the area of the corresponding rectangle in the phase plot of $\{(u(\tau),y(\tau))\ | \ \tau\in [t_1,t_1+T]\}$ which is equal to $2(\alpha-\beta)$. If it is a line then the signed-area is equal to zero. 
\end{proof}

An immediate consequence of Lemma \ref{lemma:relay_area}, is that the signed-area enclosed by the hysteresis loop of a clockwise relay operator $\mathcal{R}_{\alpha,\beta}^\circlearrowright$ defined by \eqref{eq:relay_operator_cw} is given by $-2(\alpha-\beta)$ when $u_{\min}<\beta$ and $u_{\max}>\alpha$.

We recall now the following proposition from \cite{Jayawardhana2018} which considers a particular class of Preisach operator with two-sided weighting function $\mu$ whose input-ouput behavior can exhibit butterfly loops. By two-sided we mean that there exists a simple curve $B$ that divides the Preisach domain $P$ into two disjoint subdomains $B_+$ and $B_-$ such that $P=B_+\cup B_-\cup B$ and where $\mu(\alpha,\beta)\geq 0$ for every $(\alpha,\beta)\in B_+$ and $\mu(\alpha,\beta)\leq 0$ for every $(\alpha,\beta)\in B_-$.

\begin{proposition}\label{prop:butterfly_preisach}
    ({\em \cite[Proposition 3.1]{Jayawardhana2018}}) Consider a Preisach operator $\mathcal{P}$ as in \eqref{eq:preisach_operator} with $\mu$ be a two-sided weighting function. 
    Suppose that the first order lower and upper partial moments of $\mu$ satisfy 
    \begin{align}
        \label{eq:first_moment_beta}\int_r^\infty \mu(\alpha,\beta) \beta \dd \beta & = \infty \\
        \label{eq:first_moment_alpha}\int_{-\infty}^r \mu(\alpha,\beta) \alpha \dd \alpha & = \infty,
    \end{align}
    for all $(\alpha,\beta)\in P$. Assume that the boundary curve $B$ is monotonically decreasing. Then $\mathcal{P}$ is a {\em butterfly hysteresis operator}.
\end{proposition}

\begin{proof}{}
    Let us take arbitrary $(\alpha_1,\beta_1)\in B$ with $\alpha_1>\beta_1$ (i.e., the point $(\alpha_1,\beta_1)$ is not on the boundary of $P$). Consider a subset of Preisach domain $P_1:=\{(\alpha,\beta)\in P \ |\ \alpha< \alpha_1,\, \beta> \beta_1 \}$ which is a solid triangle whose vertices are at $(\alpha_1,\beta_1)$, $(\alpha_1,\alpha_1)$ and $(\beta_1,\beta_1)$. 
    Note that since $B$ is monotonically decreasing, it separates $P$ in two polar regions where weighting function $\mu$ assigned to the domain above $B$ has different sign with that below $B$. Without loss of generality, we consider the case where $B_{1-}$ is below $B$ and $B_{1+}$ is above $B$. The arguments below are still valid when we consider the reverse case. 
    
    Due to the monotonicity of $B$, if we consider the extended area on left of $P_1$, which is given by $P_{1-}^{\text{ext}}:=\{(\alpha,\beta)\in P\ |\ \beta<\beta_1, \alpha<\alpha_1\}$, the weight $\mu$ in this area will have the same sign as that in $B_{1-}$. The same holds for the extended area above $P_1$ where the weight $\mu$ in $P_{1+}^{\text{ext}}:=\{(\alpha,\beta)\in P\ |\ \beta>\beta_1, \alpha>\alpha_1\}$ has the same sign as that in $B_{1+}$.
    
    Let us now analyze the input-output behavior when the input $u$ of the Preisach operator is a periodic signal with a period of $T$ and with one maximum $u_{\max}=\alpha_1$ and one minimum $u_{\min}=\beta_1$. 
    It is clear that for every $t\geq t_p=T$, the initial conditions of all relays in $P_1$ no longer the affect output $y$ and it becomes periodic.
    Therefore, the relays $\mathcal R^{\circlearrowleft}_{\alpha,\beta}$ whose states are switching periodically correspond to the domain $P_1$ while the state of all the relays in $P\backslash P_1$ remains the same as given by the initial condition. Consequently, following from Lemma \ref{lemma:relay_area}, the signed-area of the hysteresis loop $\mathcal{H}_{u,y}$ obtained from the input-output pair $(u,y)$ is given by
    \begin{align*}
        \mathbb A &=2\iint\displaylimits_{(\alpha,\beta)\in P_1}\mu(\alpha,\beta)\ \big(\alpha-\beta\big) \dd{\alpha} \dd\beta \\
        & = 2\underbrace{\iint\displaylimits_{(\alpha,\beta)\in B_{1+}}|\mu(\alpha,\beta)| \big(\alpha-\beta\big) \dd{\alpha} \dd\beta}_{\mathbb A_+} 
    \end{align*}
    \begin{align*}
        & \qquad - 2\underbrace{\iint\displaylimits_{(\alpha,\beta)\in B_{1-}}|\mu(\alpha,\beta)| \big(\alpha-\beta\big) \dd{\alpha} \dd\beta}_{\mathbb A_-} 
    \end{align*}
    When the variation of $\mu$ is such that $\mathbb A_+=\mathbb A_-$, we have obtained the condition for $\mathcal{P}$ to be a {\em butterfly hysteresis operator} where the chosen periodic input signal $u$ ensures that $\mathbb A=0$. However  
    since in general the variation in $\mu$ can be asymmetric, the signed-area $\mathbb A_+$ may not be equal to $\mathbb A_-$. 
    
    Let us consider the case when $\mathbb A_- > \mathbb A_+$ (i.e., the negative weight is dominant in $P_1$). In this case, we modify the periodic input signal $u$ such that its maximum $u_{\max}$ is parametrized by $\lambda > \alpha_1$. Similar as before, we have that the relays $\mathcal R^{\circlearrowleft}_{\alpha,\beta}$ whose states are switching periodically correspond to the domains $P_1$ and $P_{1+}^{\text{ext},\lambda}:=\{(\alpha,\beta)\in P\ |\ \beta>\beta_1, \lambda>\alpha>\alpha_1\}$ while the state of relays corresponding to $P\backslash (P_1\cup P_{1+}^{\text{ext},\lambda})$ remains the same as given by the initial condition. Hence, using again Lemma \ref{lemma:relay_area}, the signed-area of the hysteresis loop $\mathcal{H}_{u,y}$ corresponding to the input-output pair $(u,y)$ with modified $u$ is now given by 
    \begin{align*}
    \mathbb A &=2(\mathbb A_+ - \mathbb A_-) + 2\underbrace{\iint\displaylimits_{(\alpha,\beta)\in P_{1+}^{\text{ext},\lambda}}|\mu(\alpha,\beta)| \big(\alpha-\beta\big) \dd{\alpha} \dd\beta}_{h(\lambda)}. 
    \end{align*}
    Since $\mu$ is a piecewise continuous function, the function $h(\lambda)$ is also continuous function, $h(\alpha_1)=0$ and is strictly increasing (as $\mu>0$ in $P^{\text{ext}}_{1+}$). Due to the unboundedness of the first order upper partial moment of $\mu$ as in \eqref{eq:first_moment_beta}, it follows that $h(\lambda)\to\infty$ as $\lambda\to\infty$. This implies that there exists $\alpha_2>\alpha_1$ such that $h(\alpha_2)=\mathbb A_- - \mathbb A_+$. In this case, by taking a periodic signal with its maximum $u_{\max} = \alpha_2$ and its minimum $u_{\min}=\alpha_1$, the signed-area of the corresponding hysteresis loop $\mathcal{H}_{u,y}$ is equal to zero as claimed. 
    
    On the other hand, when $\mathbb A_+ > \mathbb A_-$, we can use {\it vis-a-vis} similar arguments as above where $u_{\min}$ is now parametrized by $\lambda < \beta_1$, instead of parameterizing $u_{\max}$ as before. For this situation, the additional relays that are affected by the modified input correspond to the domain $P_{1-}^{\text{ext},\lambda}:=\{(\alpha,\beta)\in P\ |\ \lambda<\beta<\beta_1, \alpha<\alpha_1\}$. The claim then follows similarly as above where the additional signed-area of the corresponding hysteresis loop $\mathcal{H}_{u,y}$ is a continuous function $h(\lambda)$ that is strictly increasing and approaches $-\infty$ as $\lambda \to -\infty$. 
    
    Finally, we can follow the same reasoning as above for the case when $B_{1-}$ is above $B$ and $B_{1+}$ is below $B$.
\end{proof}\vspace{0.2cm}

In Proposition \ref{prop:butterfly_preisach}, we consider a general case where $\mu$ can be any two-sided function, as long as, its decay to zero, which is measured by its upper and lower partial moments, is not too fast. If this condition is not satisfied, we may not be able to find an extended subset in $P$ parametrized by $\lambda$ (as used in the proof of Proposition \ref{prop:butterfly_preisach}) such that the total signed-area $\mathbb A$ of the hysteresis loop $\mathcal{H}_{u,y}$ with the modified $u$ is zero. Nevertheless, the conditions over the upper and lower partial moments of $\mu$ can be relaxed if we focus on a small region close to the meeting point of $B$ and the line $\{(\alpha,\beta) \ | \ \alpha = \beta \}$. This is the case for the class Preisach operator considered in the next proposition which was also introduced in \cite{Jayawardhana2018}.

\begin{proposition}\label{prop:butterfly_preisach_sym}
    ({\em \cite[Proposition 3.2]{Jayawardhana2018}}) Consider a Preisach operator $\Phi$ as in \eqref{eq:preisach_operator} with a two-sided weighting function $\mu$. Assume that the boundary curve of $\mu$ is given by $B=\{(\alpha,\beta)\in P \ | \ \alpha = -\beta + \kappa \}$ where $\kappa\in\R_+$ is an offset and $\mu$ is anti-symmetric with respect to $B$, i.e., $\mu(\alpha,\beta)=-\mu(-\alpha,-\beta)$ holds for all $(\alpha,\beta)\in P$. Then $\Phi$ is a butterfly hysteresis operator.
\end{proposition}

\begin{proof}{Proposition \ref{prop:butterfly_preisach_sym}} 
    The proof of the proposition follows similarly to the proof of Proposition \ref{prop:butterfly_preisach}. In this case, it suffices to have a periodic input signal $u$ whose maximum and minimum satisfy $u_{\max}=-u_{\min}+2\kappa$. The relays $\mathcal R_{\alpha,\beta}$ whose state switches periodically will lie in a subset of $P$ which has the form of isosceles and right triangle. Since the weighting function is anti-symmetric with respect to $B$ then the signed-area of the hysteresis loop $\mathcal{H}_{u,y}$ will be zero.   
\end{proof}

As a particular case of study, we analyze in the next example a class of Preisach butterfly operator in Proposition \ref{prop:butterfly_preisach_sym} with symmetrical two-sided weighting function.

\begin{example}\label{ex:preisach_butterfly}
    Let $B := \lbrace (\alpha,\beta)\in P\ |\ \alpha=-\beta \rbrace$ (with $\kappa=0$) 
    and consider a point $(-\beta_1,\beta_1)\in B$ such that $P_1 := \lbrace (\alpha,\beta)\in P\ |\ \alpha<\beta_1, \beta>-\beta_1 \rbrace$. In this case the subdomain of interest $P_1$ in the Preisach plane is an isosceles triangle with vertices in $(\alpha_1,\alpha_1)$, $(\alpha_1,-\alpha_1)$, and $(-\alpha_1,-\alpha_1)$. Let us define the weighting function by
    \begin{equation}\label{eq:example_butterfly_mu}
        \mu(\alpha,\beta) := \left\{ \begin{array}{ll} 
            -1 & \text{if } \alpha \leq -\beta ,\ (\alpha,\beta)\in P_1 \\
            1 & \text{if } \alpha > -\beta,\ (\alpha,\beta)\in P_1 \\
            0 & \text{otherwise }
        \end{array}\right.
    \end{equation}
    An illustration of the weighting function \eqref{eq:example_butterfly_mu} is included in Fig. \ref{fig:example_butterfly_mu}.
    \begin{figure}[ht]
        \centering
        \includegraphics[width=\columnwidth,trim={0 0 0 0},clip]{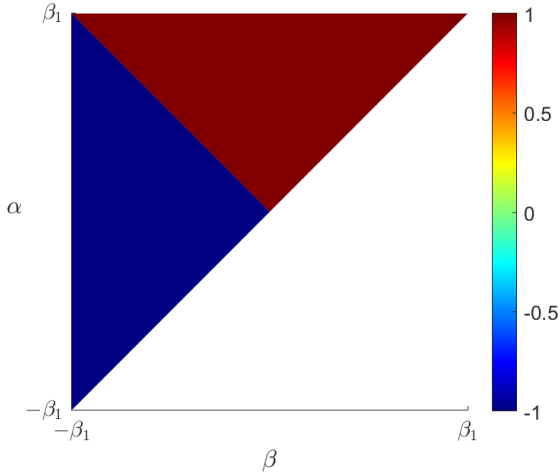}
        \caption{Weighting function $\mu(\alpha,\beta)$ defined by \eqref{eq:example_butterfly_mu}. \label{fig:example_butterfly_mu}}
    \end{figure}
    The Preisach operator with this weighting function clearly satisfies the conditions of Proposition \ref{prop:butterfly_preisach} and consequently it is a {\em butterfly hysteresis operator}. It follows from this proposition that we no longer need the extended areas $P_{1-}^{\text{ext}}$ and $P_{1+}^{\text{ext}}$. The subset of the Preisach domain $P_1$ is now subdivided in four disjoint regions defined by
    \begin{eqnarray*}
        P_{1}^{1} &:=& \lbrace (\alpha,\beta)\ |\ \alpha \leq 0,\ \beta \leq 0 \rbrace, \\
        P_{1}^{2} &:=& \lbrace (\alpha,\beta)\ |\ \alpha > 0,\ \beta \leq 0,\ \alpha \leq -\beta \rbrace, \\
        P_{1}^{3} &:=& \lbrace (\alpha,\beta)\ |\ \beta \leq 0,\ \alpha > -\beta \rbrace, \\
        P_{1}^{4} &:=& \lbrace (\alpha,\beta)\ |\ \alpha > 0,\ \beta > 0 \rbrace. \\
    \end{eqnarray*}
    \begin{figure}[ht]
        \centering
        \subfigure[Input signal]{\includegraphics[width=0.23\textwidth,trim={0 0 0 0.2cm},clip]{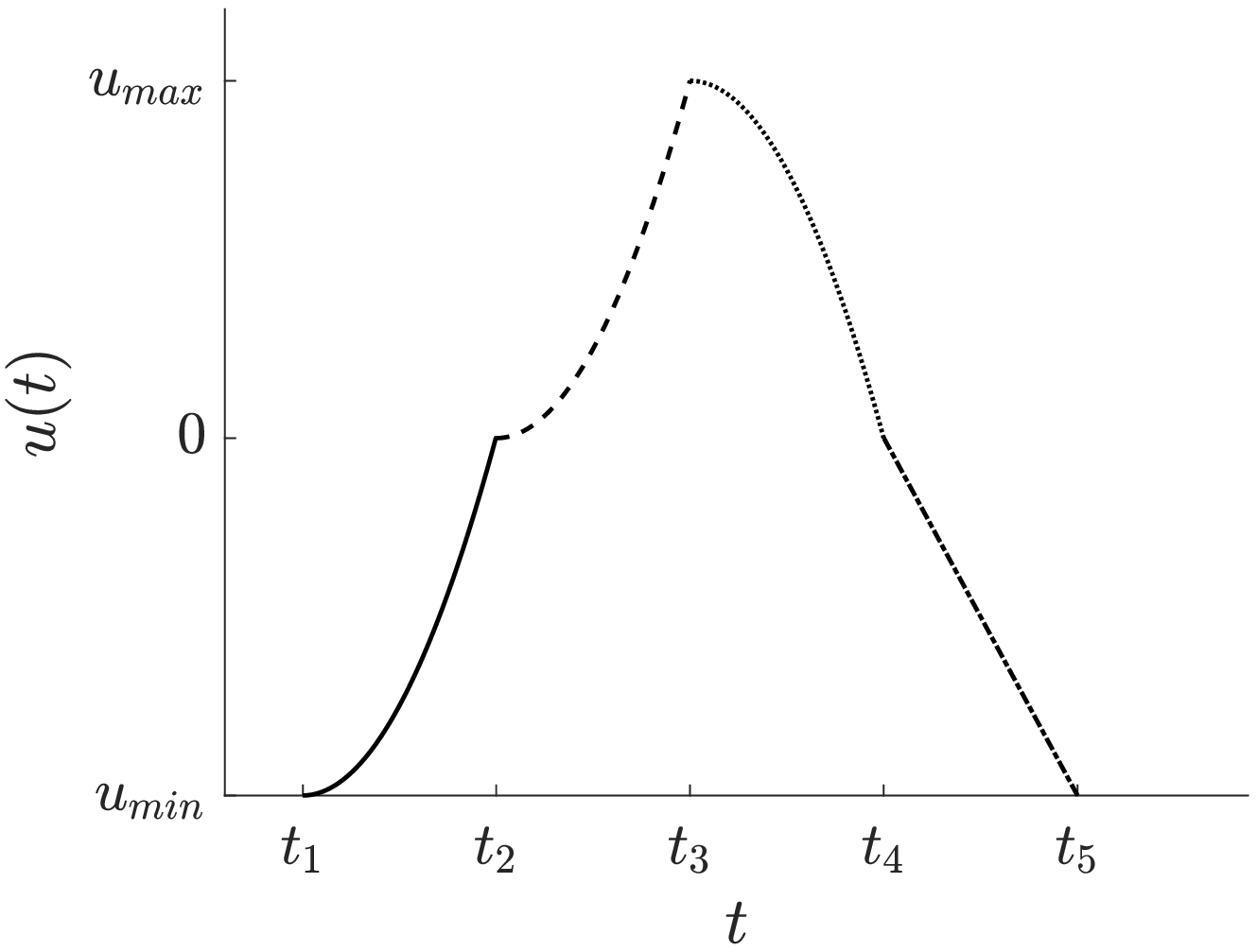}}
        \subfigure[Output signal]{\includegraphics[width=0.23\textwidth,trim={0 0 0 0.2cm},clip]{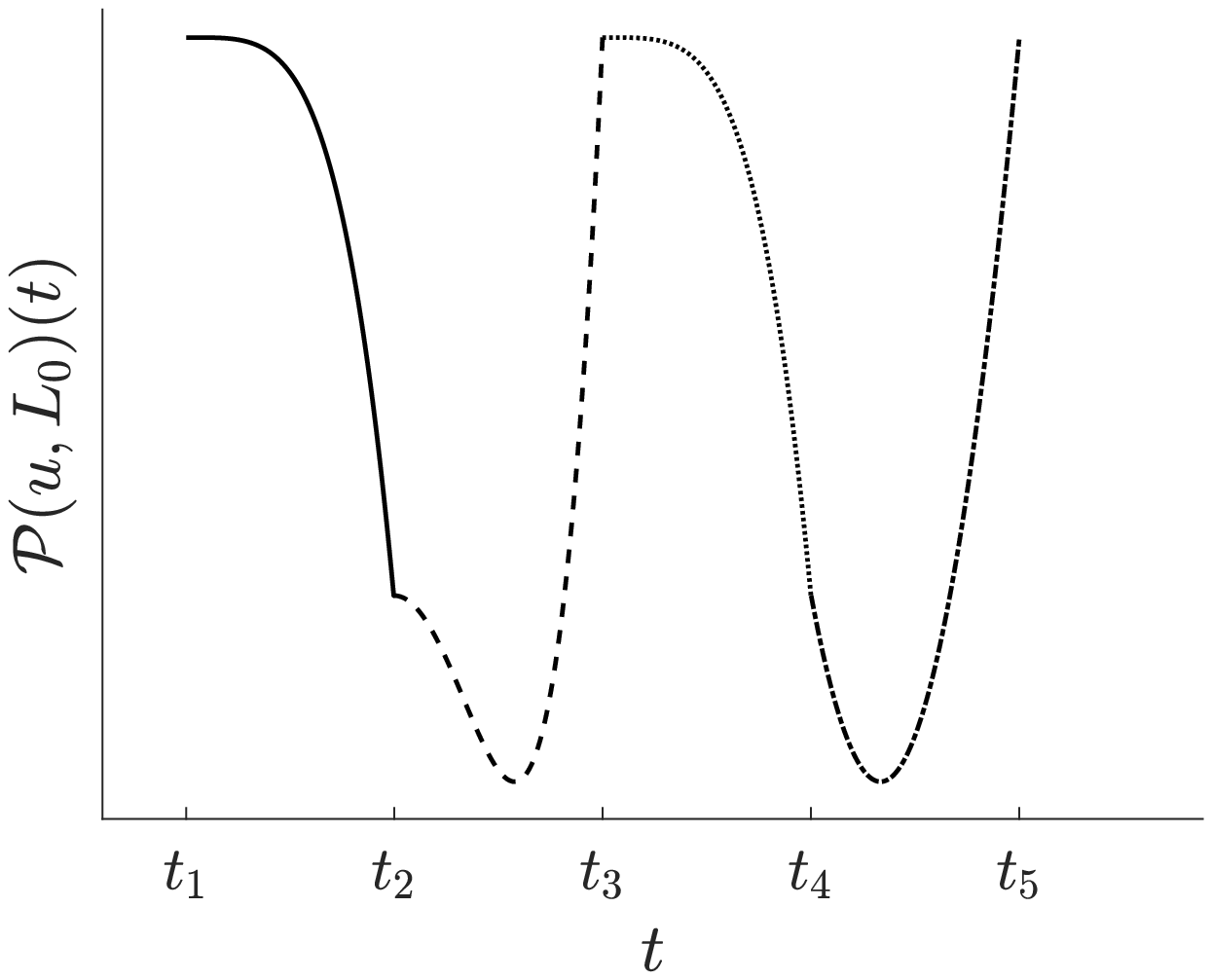}}
        \subfigure[Input-output phase plot]{\includegraphics[width=0.35\textwidth,trim={0 0 0 0.2cm},clip]{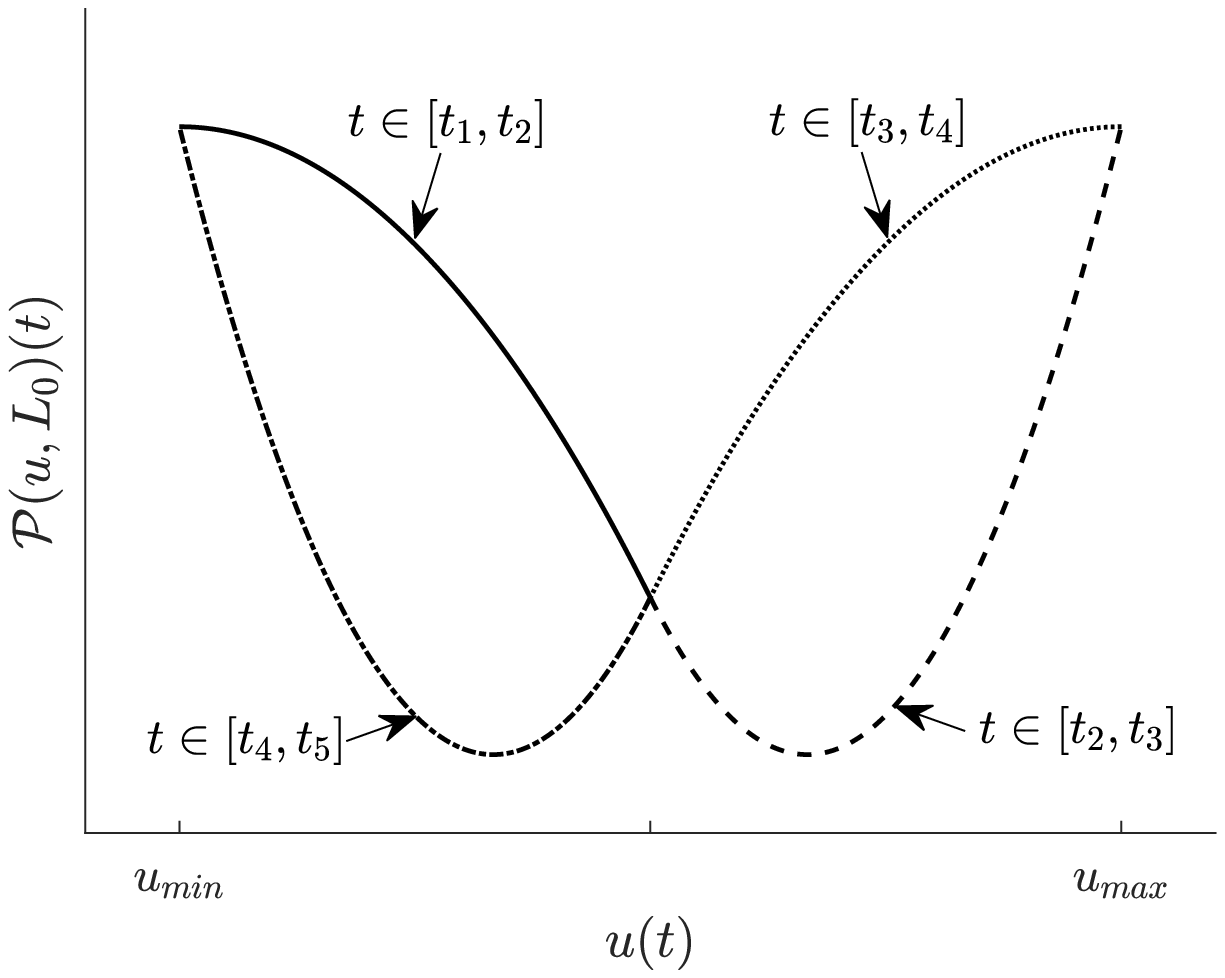}}
        \caption{Input-output phase plot using the Preisach butterfly hysteresis operator with symmetric two-sided weighting function. (a). The plot of input signal $u$ in a periodic time interval $[t_1,t_5]$. (b). The corresponding plot of output signal $y$. (c). The input-output phase plot of input and output signal which shows a symmetric butterfly loop and whose signed-area is equal to zero.}
        \label{fig:example_butterfly_phaseplot}
    \end{figure}
    Note that the output can be determined by the individual behavior of each region in the form
    \begin{equation}\label{eq:example_butterfly_y}
        \begin{aligned}
            \big(&\mathcal{P}(u,L_0)\big)(t):= \\
            &-\iintdl_{(\alpha,\beta)\in P_1^1}\ 
                \big(\mathcal
                R^{\circlearrowleft}_{\alpha,\beta}(u,r_{\alpha,\beta}(L_0)) \big)(t)\ \dd\alpha\dd\beta \\
            &-\iintdl_{(\alpha,\beta)\in P_1^2}\ 
                \big(\mathcal R^{\circlearrowleft}_{\alpha,\beta}(u,r_{\alpha,\beta}(L_0)) \big)(t)\ \dd\alpha\dd\beta \\
            &+\iintdl_{(\alpha,\beta)\in P_1^3}\ 
                \big(\mathcal R^{\circlearrowleft}_{\alpha,\beta}(u,r_{\alpha,\beta}(L_0)) \big)(t)\ \dd\alpha\dd\beta \\
            &+\iintdl_{(\alpha,\beta)\in P_1^4}\ 
                \big(\mathcal R^{\circlearrowleft}_{\alpha,\beta}(u,r_{\alpha,\beta}(L_0)) \big)(t)\ \dd\alpha\dd\beta
        \end{aligned}
    \end{equation}
    Let us analyze the input-output behavior when a periodic input $u$ with period $T$ and with one maximum $u_{\max}=\beta_1$ and one minimum $u_{\min}=-\beta_1$ is applied to this operator. Consider five time instances $t_1<t_2<t_3<t_4<t_5$ with $t_1\geq T$ and such that $u(t_1)=u_{\min}=-\beta_1$, $u(t_2)=0$, $u(t_3)=u_{\max}=\beta_1$, $u(t_4)=0$, and $u(t_5)=u_{\min}=-\beta_1$. It is clear that $u(t)$ is monotonically increasing in the interval $t_1 \leq t <t_3$ and monotonically decreasing in the interval $t_3 \leq t <t_5$. An example of an input signal $u$ satisfying these conditions is illustrated in Fig. \ref{fig:example_butterfly_phaseplot}(a). Using such input signal, the output signal $y(t)$ can be computed analytically for every $t\in[t_1,t_5]$ (i.e. for one periodic interval when the phase plot forms a hysteresis loop $\mathcal{H}_{u,y}$) based on (\ref{eq:example_butterfly_y}) and is given by
    \begin{equation}\label{eq:example_butterfly_y_solved}
        \begin{aligned}
            &\big(\mathcal{P}(u,L_0)\big)(t) = \\
                &\left\{ \begin{array}{ll} 
                -\big(u_{\max}+u(t)\big)^2&\text{if } t_1\leq t < t_2 \\ 
                -\big(u_{\max}-u(t)\big)\big(u_{\max}+3u(t)\big)&\text{if } t_2\leq t < t_3 \\ 
                -\big(u_{\max}-u(t)\big)^2&\text{if } t_3\leq t < t_4 \\ 
                -\big(u_{\max}+u(t)\big)\big(u_{\max}-3u(t)\big)&\text{if } t_4\leq t < t_5
                \end{array} \right.
        \end{aligned}
    \end{equation}
    Fig. \ref{fig:example_butterfly_phaseplot}(b) shows the corresponding output signal. By plotting the phase plot as in Fig. \ref{fig:example_butterfly_phaseplot}(c), we can see immediately that the resulting butterfly loop is symmetric as expected. Furthermore using (\ref{eq:closed_area}) it can be validated that the signed-area enclosed by this curve is equal to zero.
\end{example}\vspace*{2mm}

{
    As a final remark of this section, note that the assumption of the boundary $B$ that separates the polar regions of the weighting function being monotonically decreasing is made to simplify the analysis in Proposition \ref{prop:butterfly_preisach}. Such assumption guarantees that it is always possible to find the extended domains $P_{1+}^{\text{ext},\lambda}$ or $P_{1-}^{\text{ext},\lambda}$ where the weighting function is sign-definite.
    However, according to Definition \ref{def:butterfly_operator} and Lemma \ref{lemma:relay_area}, we have that $\mathcal{P}$ is a {\em Preisach butterfly operator} as long as we can find a {\em hysteresis loop} $\mathcal{H}_{u,y}$ with an input $u$ whose minimum and maximum can parameterize a subdomain $P_1$ of the form $P_1:=\left\{ (\alpha,\beta)\in P\ |\ u_{\min}\leq\beta\leq\alpha,\, u_{\min}\leq\alpha\leq u_{\max}\right\}$ which satisfies
    \begin{align*}
        \iint\displaylimits_{(\alpha,\beta)\in P_1}\mu(\alpha,\beta)\ \big(\alpha-\beta\big) \dd{\alpha} \dd\beta = 0.
    \end{align*}
    It follows that the monotonically decreasing property of the boundary $B$ is not necessary to obtain a {\em Preisach butterfly operator}.
}


\section{Preisach Multi-loop Operator}\label{sec:preisach_multiloop_operator}

In the previous section, we have shown that a {\em butterfly hysteresis operator} can be obtained from a Preisach operator with a two-sided weighting function.
Following from the condition of zero total signed-area in Definition \ref{def:butterfly_operator}, the previous analysis is based on finding an input $u$ such that each subloop contribution to the total signed-area of the hysteresis loop canceled each other. This analysis exploits the particular two-sided structure of the weighting function $\mu$. However, imposing this structure to the weighting function $\mu$ is only a sufficient condition to obtain a {\em butterfly hysteresis operator} which, in addition, restricts all the hysteresis loops obtained from the Preisach operator to have at most two subloops.

In this section we study a larger class of Preisach operators with more complex weighting functions that are not necessarily two-sided. The Preisach operators in this class can produce hysteresis loops with more than two subloops. Therefore, using an analysis based only on the the total enclosed signed-area of the hysteresis loops is no longer applicable for studying this class of hysteresis operators since the signed-area does not determine directly the number of subloops in a given hysteresis loop.
In this case, we must note that if a hysteresis loop has two or more subloops each one of the subloops is connected to another subloop by at least one self-intersection point of the hysteresis loop. We will call these points the crossover points of a hysteresis loop and characterize them as follows.
Consider a hysteresis loop $\mathcal{H}_{u,y}$ obtained from an input-output pair $(u,y)$ of a hysteresis operator $\Phi$ with $y=\Phi(u)$.
Let us select $t_1\geq t_p$ such that $t_1<t_2<t_1+T$ is a monotone partition of one periodic interval of $u$ and $u(t)$ is monotonically increasing when $t\in [t_1,t_2]$ and monotonically decreasing when $t \in [t_2,t_1+T]$.
We can split the hysteresis loop $\mathcal{H}_{u,y}$ into two segments that correspond to the subintervals of the monotone partition given by
\begin{align}
    \mathcal{H}_{u,y}^{+} &:= \{ (u(t),y(t))\ |\ t\in[t_1,t_2] \}, \label{eq:hysteresis_loop_partition+}\\
    \mathcal{H}_{u,y}^{-} &:= \{ (u(t),y(t))\ |\ t\in[t_2,t_1+T] \}. \label{eq:hysteresis_loop_partition-}
\end{align}
We define formally a crossover point as follows.
\begin{definition}\label{def:crossover_point}
    Consider a hysteresis loop $\mathcal{H}_{u,y}$. A point $(u_c,y_c)\in \mathcal{H}_{u,y}$ is called a {\em crossover point} if $(u_c,y_c) \in \mathcal{H}_{u,y}^{+} \cap \mathcal{H}_{u,y}^{-}$. $\hfill \triangle$
\end{definition} \vspace{0.1cm}
We remark from the definition above that a hysteresis loop $\mathcal{H}_{u,y}$ will always have at least two  crossover points corresponding to the points where the input $u$ achieves its extrema $(u_{\min},y_1), (u_{\max},y_2)\in \mathcal{H}_{u,y}$ with $(u_{\min},y_1) = (u(t_1),y(t_1)) = (u(t_1+T),y(t_1+T))$ and $(u_{\max},y_2) = (u(t_2),y(t_2))$.
Moreover, it is possible for a hysteresis loop $\mathcal{H}_{u,y}$ to have an infinite number of crossover points if, for instance, there exists a segment of intersection between $\mathcal{H}_{u,y}^{+}$ and $\mathcal{H}_{u,y}^{-}$, i.e. there exists a time subinterval $[t_3,t_4]\subset[t_1,t_1+T]$ with $t_4>t_3$ such that $(y(t),u(t))\in\mathcal{H}_{u,y}^{+} \cap \mathcal{H}_{u,y}^{-}$ for every $t\in[t_3,t_4]$.
For this reason, it can be checked that the numbers of subloops in a hysteresis loop $\mathcal{H}_{u,y}$ is not determined by the number of crossover points but by the number of maximal connected subsets in $\mathcal{H}_{u,y}^{+} \cap \mathcal{H}_{u,y}^{-}$, where each maximal connected subset can be a singleton in the case that the corresponding crossover point does not belong to a segment of intersection between $\mathcal{H}_{u,y}^{+}$ and $\mathcal{H}_{u,y}^{-}$. By a maximal connected subset in $A$ we mean a connected subset $B\subset A$ with the property that there does not exist other connected subset $C \subset A$ such that $B \subset C$. 

Using these notions we introduce the definition of {\em multi-loop hysteresis operator} as follows.
\begin{definition}\label{def:multiloop_operator}
    A hysteresis operator $\Phi$ is called a {\em multi-loop hysteresis operator} if there exists a hysteresis loop $\mathcal{H}_{u,y}$, where $y=\Phi(u)$, with at least one maximal connected subset $C\subset \mathcal{H}_{u,y}^{+}\cap\mathcal{H}_{u,y}^{-}$ such that $u_{\min}\neq u_c\neq u_{\max}$ for every $(u_c,y_c)\in C$. $\hfill \triangle$
\end{definition}\vspace{0.2cm}

Definition \ref{def:multiloop_operator} is asking for the existence of at least one maximal connected subset of crossover points besides the ones that contain the crossover points corresponding to the maximum and minimum values of the input. The existence of this maximal connected subset guarantees that the hysteresis loop will composed of at least two subloops. Moreover, it is clear that every {\em butterfly hysteresis operator} is then a {\em multi-loop hysteresis operator}.
Before characterizing the class of Preisach operators that satisfy the conditions to be {\em multi-loop hysteresis operators}, we introduce next lemma that allows us to relate the existence of a crossover point in a hysteresis loop obtained from a Preisach operator with the integration of its weighting function $\mu$ over a rectangular region of $P$ delimited by the maximum and minimum values of the input.

\begin{lemma}\label{lemma:crossover_point}
    Consider a hysteresis loop $\mathcal{H}_{u,y}$ obtained from an input-output pair $(u,y)$ of a Preisach operator $\mathcal{P}$ with a weighting function $\mu$. A point $(u_c, y_c)\in\mathcal{H}_{u,y}$ is a {\em crossover point} if and only if
    \begin{equation}\label{eq:crossover_condition}
        \iint\displaylimits_{(\alpha,\beta)\in \Omega_c} \mu(\alpha,\beta)\ \dd\alpha \dd\beta = 0
    \end{equation}
    where the region $\Omega_c$ is defined by $\Omega_c:=\{(\alpha,\beta)\in P\ |\ u_c<\alpha< u_{\max},\ u_{\min}<\beta< u_c \}$.
\end{lemma}\vspace*{0.2cm}

\begin{proof}{Lemma \ref{lemma:crossover_point}}
\noindent({\em Sufficiency})\ 
Let $(u_c,y_c)$ be the {\em crossover point} of a hysteresis loop $\mathcal{H}_{u,y}$ and consider the corresponding subsets $\mathcal{H}_{u,y}^{+}$ and $\mathcal{H}_{u,y}^{-}$ of $\mathcal{H}_{u,y}$ as defined in \eqref{eq:hysteresis_loop_partition+} and \eqref{eq:hysteresis_loop_partition-} with a monotone partition $t_1<t_2<t_1+T$ of the periodic input $u$. 
When $u_c=u_{\min}=u(t_1)=u(t_1+T)$ or $u_c=u_{\max}=u(t_2)$, the region $\Omega_c$ is empty and \eqref{eq:crossover_condition} holds trivially.
Therefore, let us consider the case when there exist two time instants $\tau_1 \in (t_1, t_2)$ and $\tau_2 \in (t_2, t_1 + T)$ such that $(u_c,y_c)=(u(\tau_1),y(\tau_1))=(u(\tau_2),y(\tau_2))$ with $u_{\min}\neq u_c\neq u_{\max}$.  

Let us analyze the input-output behavior of the Preisach operator in the intervals $(t_1, t_2)$ and $(t_1, t_2+T)$. For this, consider a subdomain of the Preisach plane given by $P_1:=\{(\alpha,\beta)\in P\ |\ \alpha< u_{\max},\ \beta> u_{\min} \}$ which is a triangle whose vertices are at $(u_{\max},u_{\min})$, $(u_{\max},u_{\max})$ and $(u_{\min},u_{\min})$. It is clear that at every time instance $t\geq t_p=T$ the state of relays in $P\backslash P_1$ remains the same as given by the initial condition. We define three time varying disjoint regions of $P_1$ whose boundaries depend on the instantaneous value of the input $u(t)$ and which are given by
\begin{align}
	\Omega_1(t) &:= \left\{ (\alpha,\beta) \in P_1\ |\ u(t)<\beta \right\},\label{eq:crossover_timevar_omega1}\\
	\Omega_2(t) &:= \left\{ (\alpha,\beta) \in P_1\ |\ \beta<u(t)<\alpha \right\},\label{eq:crossover_timevar_omega2}\\
	\Omega_3(t) &:= \left\{ (\alpha,\beta) \in P_1\ |\ \alpha<u(t) \right\}.\label{eq:crossover_timevar_omega3}
\end{align}
The region $\Omega_1(t)$ is a triangle whose vertices are at $(u(t),u(t))$, $(u_{\max},u(t))$ and $(u_{\max},u_{\max})$, the region $\Omega_3(t)$ is a triangle whose vertices are at $(u(t),u(t))$, $(u(t),u_{\min})$ and $(u_{\min},u_{\min})$,
and the region $\Omega_2(t)$ is a rectangle whose vertices are at $(u(t),u(t))$, $(u_{\max},u_{\min})$, $(u(t),u_{\min})$ and $(u_{\max},u(t))$. 
It can be checked that for every time instance $t\in[t_1,t_1+T]$, all relays corresponding to the regions $\Omega_1(t)$ and $\Omega_3(t)$ are in state $-1$ and $+1$, respectively. Moreover, at time instances $\tau_1$ and $\tau_2$ we have that  $\Omega_2(\tau_1)=\Omega_2(\tau_2)=\Omega_c$.

The required condition \eqref{eq:crossover_condition} is obtained computing the output of the Preisach operator at both time instances $\tau_1$ and $\tau_2$ using the regions $\Omega_1(t)$, $\Omega_2(t)$ and $\Omega_3(t)$ as follows. At time instance $t_1$ the input $u$ reaches its minimum value. Thus we have $u(t_1) = u_{\min}$ which implies that all relays in the subdomain $P_1$ are in $-1$ state. As the input increases, at every time instance $t\in(t_1,t_2)$ the region $\Omega_3(t)$ indicates the relays whose state has changed from $-1$ to $+1$ while in the regions $\Omega_2(t)$ and $\Omega_1(t)$ all relays remain in $-1$ state. Therefore, at time instance $\tau_1$ the output of the Preisach operator is given by
\begin{equation}\label{eq:cross_output_input_incr}
	\begin{aligned}
		& y(\tau_1) = \big(\mathcal{P}(u,L_0)\big)(\tau_1)\\[1mm]
				&-\disp\iint\displaylimits_{(\alpha,\beta)\in \Omega_1(\tau_1)}\ \mu(\alpha, \beta) \ \dd\alpha \dd\beta \ \
				-\disp\iint\displaylimits_{(\alpha,\beta)\in \Omega_2(\tau_1)}\ \mu(\alpha, \beta) \ \dd\alpha \dd\beta \\
				&+\disp\iint\displaylimits_{(\alpha,\beta)\in \Omega_3(\tau_1)}\ \mu(\alpha, \beta) \ \dd\alpha \dd\beta \\
				&+\disp\iint\displaylimits_{(\alpha,\beta)\in P\backslash P_1}\ \mu(\alpha, \beta) \big(\mathcal R^{\circlearrowleft}_{\alpha,\beta}(u,r_{\alpha,\beta}(L_0)) \big)(\tau_1)\ \dd\alpha \dd\beta.
	\end{aligned}
\end{equation}
At time instance $t_2$ the input $u$ reaches now its maximum value. Thus, in this case we have $u(t_2) = u_{\max}$ which implies that all relays in the subdomain $P_1$ are in $+1$ state. Subsequently, as the input decreases, at every time instance $t\in(t_2,t_1+T)$ the region $\Omega_1(t)$ indicates the relays whose state has changed from $+1$ to $-1$ while in the regions $\Omega_2(t)$ and $\Omega_3(t)$ all relays remain in $+1$ state. Therefore, at time instance $\tau_2$ the output of the Preisach operator is given by
\begin{equation}\label{eq:cross_output_input_decr}
	\begin{aligned}
		& y(\tau_2) = \big(\mathcal{P}(u,L_0)\big)(\tau_2)\\[1mm]
				&-\disp\iint\displaylimits_{(\alpha,\beta)\in \Omega_1(\tau_2)}\ \mu(\alpha, \beta) \ \dd\alpha \dd\beta \ \
				+\disp\iint\displaylimits_{(\alpha,\beta)\in \Omega_2(\tau_2)}\ \mu(\alpha, \beta) \ \dd\alpha \dd\beta \\
				&+\disp\iint\displaylimits_{(\alpha,\beta)\in \Omega_3(\tau_2)}\ \mu(\alpha, \beta) \ \dd\alpha \dd\beta \\
				&+\disp\iint\displaylimits_{(\alpha,\beta)\in P\backslash P_1}\ \mu(\alpha, \beta) \big(\mathcal R^{\circlearrowleft}_{\alpha,\beta}(u,r_{\alpha,\beta}(L_0)) \big)(\tau_2)\ \dd\alpha \dd\beta.
	\end{aligned}
\end{equation}
Subtracting (\ref{eq:cross_output_input_decr}) and (\ref{eq:cross_output_input_incr}) we have
\begin{equation}\label{eq:cross_output_diff}
	\begin{aligned}
		0 &= y(\tau_2) - y(\tau_1) \\
        &= \disp\iint\displaylimits_{(\alpha,\beta)\in \Omega_2(\tau_2)}\ \mu(\alpha, \beta)\ \dd\alpha \dd\beta 
        + \disp\iint\displaylimits_{(\alpha,\beta)\in \Omega_2(\tau_1)}\ \mu(\alpha, \beta)\ \dd\alpha \dd\beta\\
		&= 2\disp\iint\displaylimits_{(\alpha,\beta)\in \Omega_c}\ \mu(\alpha, \beta)\ \dd\alpha \dd\beta.
	\end{aligned}
\end{equation}

\noindent({\em Necessity})\ Assume that \eqref{eq:crossover_condition} 
holds for some value $u_c\in[u_{\min},u_{\max}]$.
Consider the time instance $\tau_1 \in [t_1, t_2]$ when $u(\tau_1)=u_c$. At this time instance, the output $y(\tau_1)$ is given as in (\ref{eq:cross_output_input_incr}) and $(u(\tau_1),y(\tau_1))\in \mathcal{H}_{u,y}^{+}$. Similarly, let $\tau_2 \in [t_2, t_1+T]$ be the time instance when $u(\tau_2)=u_c$. At this time instance the output $y(\tau_2)$ is given as in (\ref{eq:cross_output_input_decr}) and $(u(\tau_2),y(\tau_2))\in \mathcal{H}_{u,y}^{-}$. Since \eqref{eq:crossover_condition} holds, by subtracting (\ref{eq:cross_output_input_decr}) and (\ref{eq:cross_output_input_incr}) we obtain (\ref{eq:cross_output_diff}) again. It follows that $y(\tau_2) = y(\tau_1) = y_c$, and consequently $(u_c,y_c)\in \mathcal{H}_{u,y}^{+} \cap \mathcal{H}_{u,y}^{-}$.
\end{proof}


It is clear that Lemma \ref{lemma:crossover_point} can be used to find the crossover points of a hysteresis loop obtained from a Preisach operator. However, noting that the region $\Omega_c$ in \eqref{eq:crossover_condition} depends explicitly on the maximum and minimum of the input $u$ applied to the Preisach operator, we could also use Lemma \ref{lemma:crossover_point} to estimate an input $u$ that produces a hysteresis loop with crossover points additional to the trivial ones corresponding to the maximum and minimum value of the input. 

\begin{example}\label{ex:crossoverpoints_butterfly_and_double_loop}
    Let us recall the Preisach operator from Example \ref{ex:preisach_butterfly} whose weighting function $\mu$ is defined in \eqref{eq:example_butterfly_mu}. We can check that a non-empty region $\Omega_c$ satisfying \eqref{eq:crossover_condition} is given by
    \begin{equation}\label{eq:example_doubleloop_omegac}
        \Omega_c = \{(\alpha,\beta)\in P\ |\ 0<\alpha<\beta_1,\ -\beta_1<\beta<0 \}, 
    \end{equation}
    which is illustrated in Fig. \ref{fig:example_butterfly_mu_omegac}. 
    Therefore, noting the limits of region $\Omega_c$ in \eqref{eq:example_doubleloop_omegac}, it follows that applying to this Preisach operator an input $u$ with one maximum $u_{\max}=\beta_1$ and one minimum $u_{\min}=-\beta_1$ yields a crossover point $(u_c,y_c)$ with $u_c=0$ as it has been shown in the phase plot of Fig \ref{fig:example_butterfly_phaseplot}. 
    Furthermore, using \eqref{eq:example_butterfly_y_solved} we can find that $y_c=-u_{\max}^2$.
    \begin{figure}[ht]
        \centering
        \includegraphics[width=\columnwidth,trim={0 0 0 0},clip]{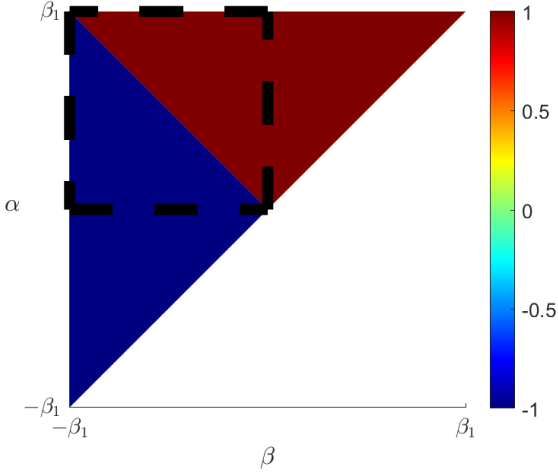}
        \caption{Weighting function $\mu(\alpha,\beta)$ of Example \ref{ex:preisach_butterfly} defined in \eqref{eq:example_butterfly_mu} where the region $\Omega_c$ given by \eqref{eq:example_butterfly_mu} and that satisfies \eqref{eq:crossover_condition} is indicated by the dashed line. \label{fig:example_butterfly_mu_omegac}}
    \end{figure}
\end{example}\vspace{0.1cm}

We remark that only the existence of crossover points does not guarantee that the hysteresis loop is composed of subloops with different orientation. We illustrate this with the next example.

\begin{example}\label{ex:preisach_multiloop}
    Consider a subdomain of Preisach plane $P_1 = \lbrace (\alpha,\beta)\in P\ |\ \alpha<\beta_1,\,\beta>-\beta_1 \rbrace$ with $\beta_1>0$ and define a weighting function $\mu$ given by
    \begin{equation}\label{eq:example_doubleloop_mu}
        \begin{aligned}
            \mu(\alpha,\beta) &:= \left\{\begin{array}{rl}
                -1, & \text{if } (\alpha,\beta) \in P_{1-}, \\
                1, & \text{if } (\alpha,\beta) \in P_1\backslash P_{1-},\\
                0, & \text{otherwise}.
            \end{array}\right.
        \end{aligned}
    \end{equation}
    where $P_{1-}=\{ (\alpha,\beta)\in P_1\ |\ -\beta_1<\beta<0,\,0<\alpha<\beta+\beta_1 \}$.
    It can be checked that the same region $\Omega_c$ given as in \eqref{eq:example_doubleloop_omegac} satisfies condition \eqref{eq:crossover_condition} with the weighting function $\mu$ defined by \eqref{eq:example_doubleloop_mu}. Fig. \ref{fig:doubleloop_same_oriented_mu} illustrates this weighting function with the region $\Omega_c$ indicated by a dashed line. It follows that the hysteresis loop obtained from a Preisach operator with a weighting function $\mu$ defined by \eqref{eq:example_doubleloop_mu} and whose input $u$ has one maximum $u_{\max}=\beta_1$ and one minimum $u_{\min}=-\beta_1$ has a crossover point with coordinates $(u_c,y_c)=(0,0)$, where $y_c$ is computed using \eqref{eq:cross_output_input_incr} or \eqref{eq:cross_output_input_decr}. Nevertheless, by simple inspection of the phase plot include in Fig. \ref{fig:doubleloop_same_oriented_phaseplot}, we can check that the hysteresis loop is composed of two subloops with the same orientation.
    
    \begin{figure}[ht]
        \centering
        \includegraphics[width=\columnwidth,trim={0 0 0 0},clip]{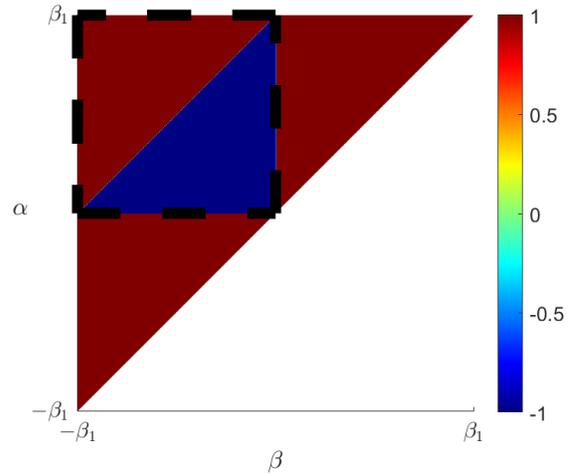}
        \caption{Weighting function $\mu(\alpha,\beta)$ of a Preisach operator which can produce a hysteresis loop whose subloops have the same orientation. The region $\Omega_c$ that satisfies \eqref{eq:crossover_condition} is indicated by the dashed line. \label{fig:doubleloop_same_oriented_mu}}
    \end{figure}
    
    \begin{figure}[ht]
        \centering
        \subfigure[Input signal]{\includegraphics[width=0.23\textwidth,trim={0 0 0 0.2cm},clip]{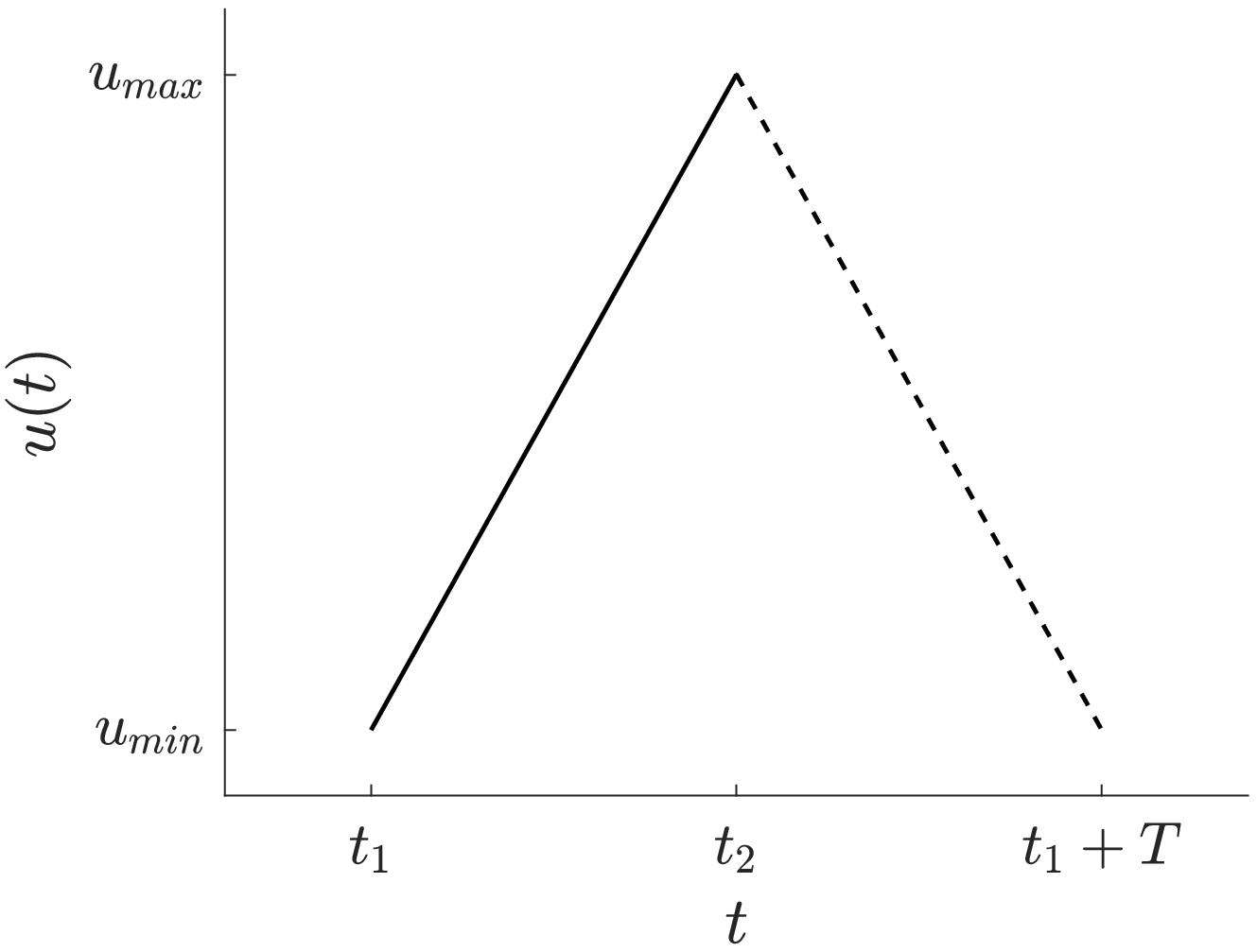}}
        \subfigure[Output signal]{\includegraphics[width=0.23\textwidth,trim={0 0 0 0.2cm},clip]{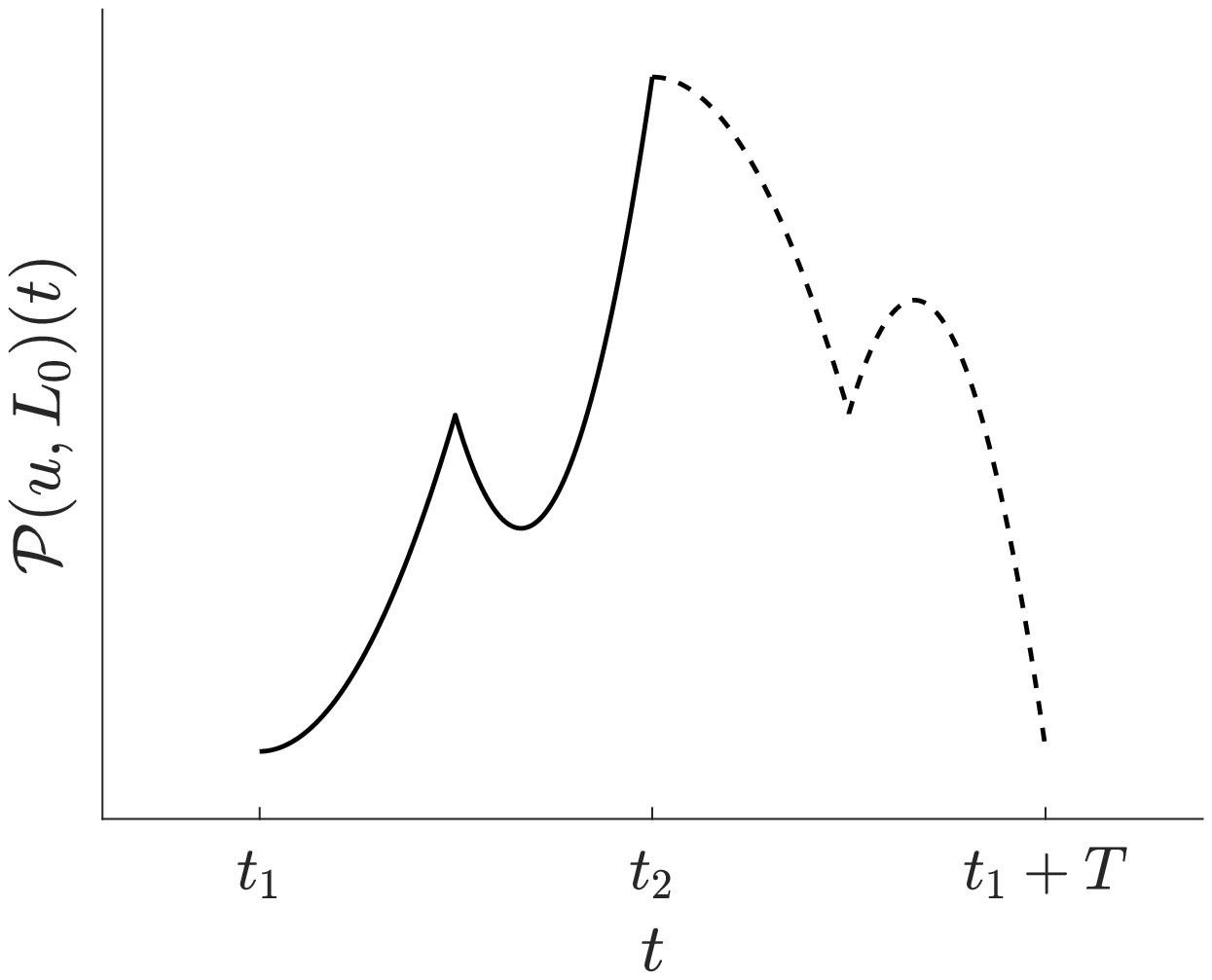}}
        \subfigure[Input-output phase plot]{\includegraphics[width=0.35\textwidth,trim={0 0 0 0.2cm},clip]{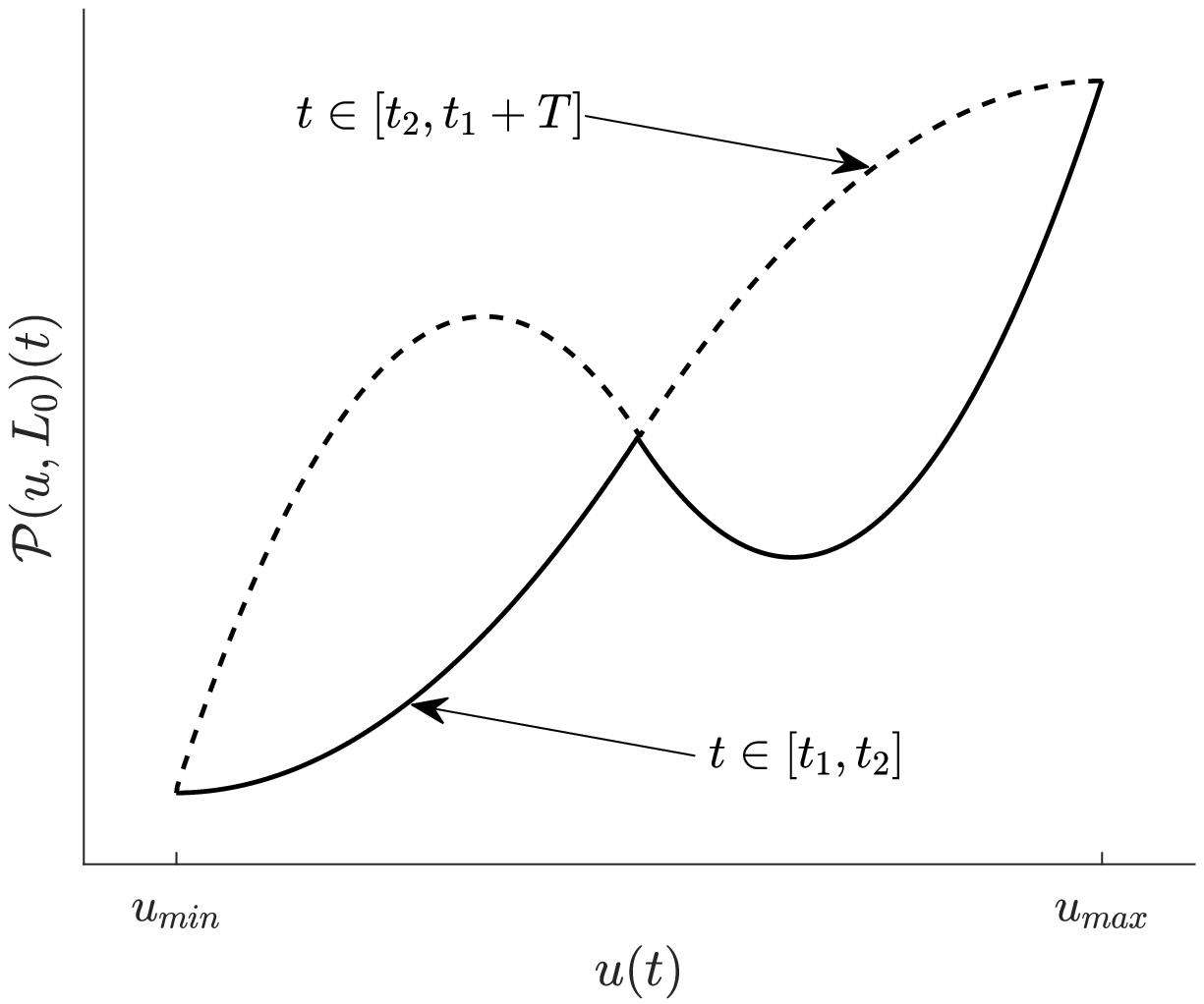}}
        \caption{Input-output hysteresis response using the multiple loops Preisach multi-loop operator. (a). The plot of input signal $u$ in a periodic time interval $[t_1,t_1+T]$. (b). The corresponding plot of output signal $y$. (c). The input-output phase plot of input and output signal which shows hysteresis loop with two subloops in the same orientation.}
        \label{fig:doubleloop_same_oriented_phaseplot}
    \end{figure}
\end{example}\vspace{0.1cm}


We introduce now a proposition that shows how a {\em multi-loop hysteresis operator} can be obtained from a Preisach operator.

\begin{proposition}\label{prop:multiloop_preisach_operator}
    Consider a Preisach operator $\mathcal{P}$ as in $\eqref{eq:preisach_operator}$ with a weighting function $\mu$. Assume that there exists a point $(\alpha_0,\beta_0)\in P$ with $\alpha_0>\beta_0$ and three values $\alpha_{1^-}<\alpha_{1}<\alpha_{1^+}$ such that $\beta_0<\alpha_{1^-}$ and $\alpha_{1^+}<\alpha_0$, and \eqref{eq:crossover_condition} holds for the region
    \begin{equation}\label{eq:multiloop_preisach_omegac}
        \Omega_{\alpha_1}=\{(\alpha,\beta)\in P\ |\ \alpha_1<\alpha<\alpha_0,\ \beta_0<\beta<\alpha_1 \}
    \end{equation}
    but does not hold for the regions
    \begin{align}
        \Omega_{\alpha_{1^-}}&=\{(\alpha,\beta)\in P\ |\ \alpha_{1^-}<\alpha<\alpha_0,\ \beta_0<\beta<\alpha_{1^-} \}, \label{eq:multiloop_preisach_not_omegac1}\\
        \Omega_{\alpha_{1^+}}&=\{(\alpha,\beta)\in P\ |\ \alpha_{1^+}<\alpha<\alpha_0,\ \beta_0<\beta<\alpha_{1^+} \} \label{eq:multiloop_preisach_not_omegac2}.
    \end{align}
    Then $\mathcal{P}$ is a {\em multi-loop hysteresis operator}.
\end{proposition}\vspace{0.2cm}

\begin{proof}{Proposition \ref{prop:multiloop_preisach_operator}}
    Consider the hysteresis loop $\mathcal{H}_{u,y}$ obtained from the input-output pair $(u,y)$ with the input $u$ being periodic with one maximum $u_{\max} = \alpha_0$ and one minimum $u_{\min} = \beta_0$, and $y=\mathcal{P}(u,L_0)$.
    Let $t_1<t_2<t_1+T$ be the monotonic partition of the input that divides $\mathcal{H}_{u,y}$ into $\mathcal{H}_{u,y}^{-}$ and $\mathcal{H}_{u,y}^{+}$ with $u(t_1)=u(t_1+T)=u_{\min}$ and $u(t_2)=u_{\max}$, and consider six time instances $\tau_{1^-},\tau_1,\tau_{1^+}\in[t_1,t_2]$ and $\tau_{2^-},\tau_2,\tau_{2^+}\in[t_2,t_1+T]$ with $\tau_{1^-}<\tau_1<\tau_{1^+}$ and $\tau_{2^-}>\tau_2>\tau_{2^+}$ such that $u(\tau_{1^-})=u(\tau_{2^-})=\alpha_{1^-}$, $u(\tau_{1})=u(\tau_{2})=\alpha_{1}$ and $u(\tau_{1^+})=u(\tau_{2^+})=\alpha_{1^+}$.
    
    Using Lemma \ref{lemma:crossover_point} with the region $\Omega_{\alpha_{1}}$ defined in \eqref{eq:multiloop_preisach_omegac}, the hysteresis loop $\mathcal{H}_{u,y}$ has a crossover point $(u_c,y_c)$ with $u_c=u(\tau_{1})=u(\tau_{2})=\alpha_1$ and where $y_c=y(\tau_{1})=y(\tau_{2})$ is given by \eqref{eq:cross_output_input_incr} or \eqref{eq:cross_output_input_decr}. 
    
    Without loss of generality, let $C$ be the maximal connected subset of $\mathcal{H}_{u,y}^{+}\cap\mathcal{H}_{u,y}^{-}$ that contains $(u_c,y_c)$. To check that $C$ does not contain a crossover point of the form $(u_{\min},y_1)$ observe that since \eqref{eq:crossover_condition} does not hold for the region $\Omega_{\alpha_{1^-}}$ then using again Lemma \ref{lemma:crossover_point} we have that $(u(\tau_{1^-}),y(\tau_{1^-}))$ and $(u(\tau_{2^-}),y(\tau_{2^-}))$ are not crossover points and are not in $\mathcal{H}_{u,y}^{+}\cap\mathcal{H}_{u,y}^{-}$. Consequently, there does not exist connected subset of $\mathcal{H}_{u,y}^{+}\cap\mathcal{H}_{u,y}^{-}$ that could contain both $(u_c,y_c)$ and $(u_{\min},y_1)$.
    Similarly, we can check that $C$ does not contain a crossover point of the form $(u_{\max},y_2)$ by noting that \eqref{eq:crossover_condition} does not hold for the region $\Omega_{\alpha_{1^+}}$ which by Lemma \ref{lemma:crossover_point} implies that $(u(\tau_{1^+}),y(\tau_{1^+}))$ and $(u(\tau_{2^+}),y(\tau_{2^+}))$ are not crossover points and are not in $\mathcal{H}_{u,y}^{+}\cap\mathcal{H}_{u,y}^{-}$. It follows again that
    there does not exist connected subset of $\mathcal{H}_{u,y}^{+}\cap\mathcal{H}_{u,y}^{-}$ that could contain both $(u_c,y_c)$ and $(u_{\max},y_2)$.
\end{proof}\vspace{0.1cm}

We remark that Proposition \ref{prop:multiloop_preisach_operator} could be extended to consider more than one region $\Omega_\alpha$ where \eqref{eq:crossover_condition} holds. Let $\mu$ be a weighting function and consider values
\begin{equation*}
    \alpha_{i^-} <\alpha_i < \alpha_{i^+}, \quad \text{where } i\in \left\{ 1,\dots,m \right\} \text{ and } m\in \mathbb{Z}_+,
\end{equation*}
with $\beta_0<\alpha_{1^-}$ and $\alpha_{m^+}<\alpha_0$, and such that for every $j\in \left\{1,\dots,m-1 \right\}$ we have that $\alpha_{j^+}<\alpha_{(j+1)^-}$. Assume that using these values, we can construct regions given by
\begin{equation*}
    \Omega_{\alpha_i}=\{(\alpha,\beta)\in P\ |\ \alpha_i<\alpha<\alpha_0,\ \beta_0<\beta<\alpha_i \},
\end{equation*}
such that \eqref{eq:crossover_condition} holds but does not hold for regions given by
\begin{align*}
    \Omega_{\alpha_{i^-}}&=\{(\alpha,\beta)\in P\ |\ \alpha_{i^-}<\alpha<\alpha_0,\ \beta_0<\beta<\alpha_{i^-} \},\\
    \Omega_{\alpha_{i^+}}&=\{(\alpha,\beta)\in P\ |\ \alpha_{i^+}<\alpha<\alpha_0,\ \beta_0<\beta<\alpha_{i^+} \}
\end{align*}
for every $i\in \left\{ 1,\dots,m \right\}$. It follows immediately from Proposition \ref{prop:multiloop_preisach_operator} that a Preisach operator with this weighting function is a {\em multi-loop hysteresis operator}. Moreover, it can be checked that in this case the hysteresis loop $\mathcal{H}_{u,y}$ obtained from such Preisach operator with an input whose maximum is $u_{\max} = \alpha_0$ and minimum is $u_{\min}=\beta_0$ will be composed of $m+1$ subloops. 
The final example of this section illustrates a Preisach operator with a weighting function that has a complex distribution of positive and negative domains and whose hysteresis loops has four subloops.


\begin{example}\label{ex:multiloop_butterfly}
    Consider a subset of Preisach domain $P_1:=\left\{ (\alpha,\beta)\ |\ -1<\beta<1,\ \beta<\alpha<1 \right\}$ and a weighting function defined by
    \begin{equation}\label{eq:example_multiloop_mu}
        \begin{aligned}
            &\mu(\alpha,\beta) := \\[0.1cm]
            &\left\{\begin{array}{c@{\hspace{0.25cm}}l}
                \sin\left( 2\pi \left( \alpha-\beta \right) \right) + \sin\left( 2\pi \left( \alpha+\beta \right) \right), & \text{if } (\alpha,\beta) \in P_1, \\
                0, & \text{otherwise}.
            \end{array}\right.
        \end{aligned}
    \end{equation}
    
    \begin{figure}[ht]
        \centering
        \includegraphics[width=\columnwidth,trim={0 0 0 0},clip]{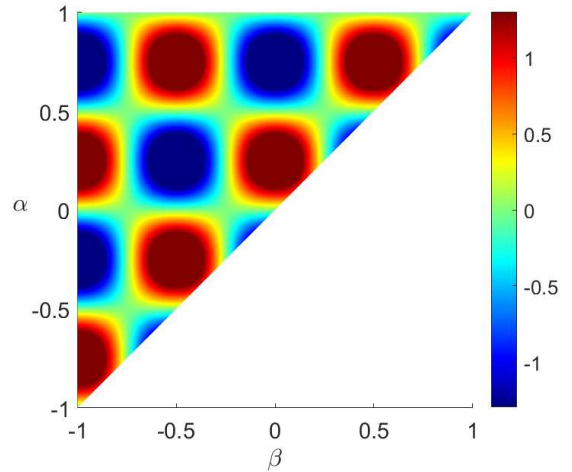}
        \caption{Weighting function $\mu(\alpha,\beta)$ defined in \eqref{eq:example_multiloop_mu} corresponding to a {\em Preisach multi-loop operator}. \label{fig:example_multiloop_mu}}
    \end{figure}
    
    The weighting function $\mu$ defined in \eqref{eq:example_multiloop_mu} is illustrated in Fig \ref{fig:example_multiloop_mu}. For this weighting function there exist three non-empty regions $\Omega_{\alpha_1}$, $\Omega_{\alpha_2}$ and $\Omega_{\alpha_3}$ that satisfy \eqref{eq:crossover_condition} and which are given by
    \begin{equation}\label{eq:example_multiloop_omegac}
        \begin{aligned}
            \Omega_{\alpha_1} &= \{(\alpha,\beta)\in P_1\ |\ -0.5<\alpha<1,\ -1<\beta<-0.5 \}, \\
            \Omega_{\alpha_2} &= \{(\alpha,\beta)\in P_1\ |\ 0<\alpha<1,\ -1<\beta<0 \}, \\
            \Omega_{\alpha_3} &= \{(\alpha,\beta)\in P_1\ |\ 0.5<\alpha<1,\ -1<\beta<0.5 \}. 
        \end{aligned}
    \end{equation}
    Moreover, it can be verified that \eqref{eq:crossover_condition} does not hold for every of next the regions
    \begin{equation*}
        \begin{aligned}
            \Omega_{\alpha_{1^-}} &= \{(\alpha,\beta)\in P_1\, |\, -0.75<\alpha<1, -1<\beta<-0.75 \}, \\
            \Omega_{\alpha_{1^+}} &= \{(\alpha,\beta)\in P_1\, |\, -0.25<\alpha<1, -1<\beta<-0.25 \}, \\
            \Omega_{\alpha_{2^-}} &= \{(\alpha,\beta)\in P_1\, |\, -0.25<\alpha<1, -1<\beta<-0.25 \}, \\
            \Omega_{\alpha_{2^+}} &= \{(\alpha,\beta)\in P_1\, |\,  0.25<\alpha<1, -1<\beta<0.25 \}, \\
            \Omega_{\alpha_{3^-}} &= \{(\alpha,\beta)\in P_1\, |\,  0.25<\alpha<1, -1<\beta<0.25 \}, \\
            \Omega_{\alpha_{3^+}} &= \{(\alpha,\beta)\in P_1\, |\,  0.75<\alpha<1, -1<\beta<0.75 \}.
        \end{aligned}
    \end{equation*}
    
    Fig. \ref{fig:example_multiloop_omegac} illustrates the three regions $\Omega_{\alpha_1}$, $\Omega_{\alpha_2}$ and $\Omega_{\alpha_3}$ with a dashed line and Fig. \ref{fig:example_multiloop_phaseplot} shows the input-output phase plot of the Preisach operator with the weighting function \eqref{eq:example_multiloop_mu} with a periodic input whose maximum and minimum are $u_{\min}=-1$ and $u_{\max}=1$. It can be verified that the hysteresis loop is composed of four subloops and that there exist three crossover points additional to the trivial ones corresponding to the maximum and minimum of the input.
    
    \begin{figure}[ht]
        \centering
        \subfigure[Input signal]{\includegraphics[width=0.23\textwidth,trim={0 0 0 0.2cm},clip]{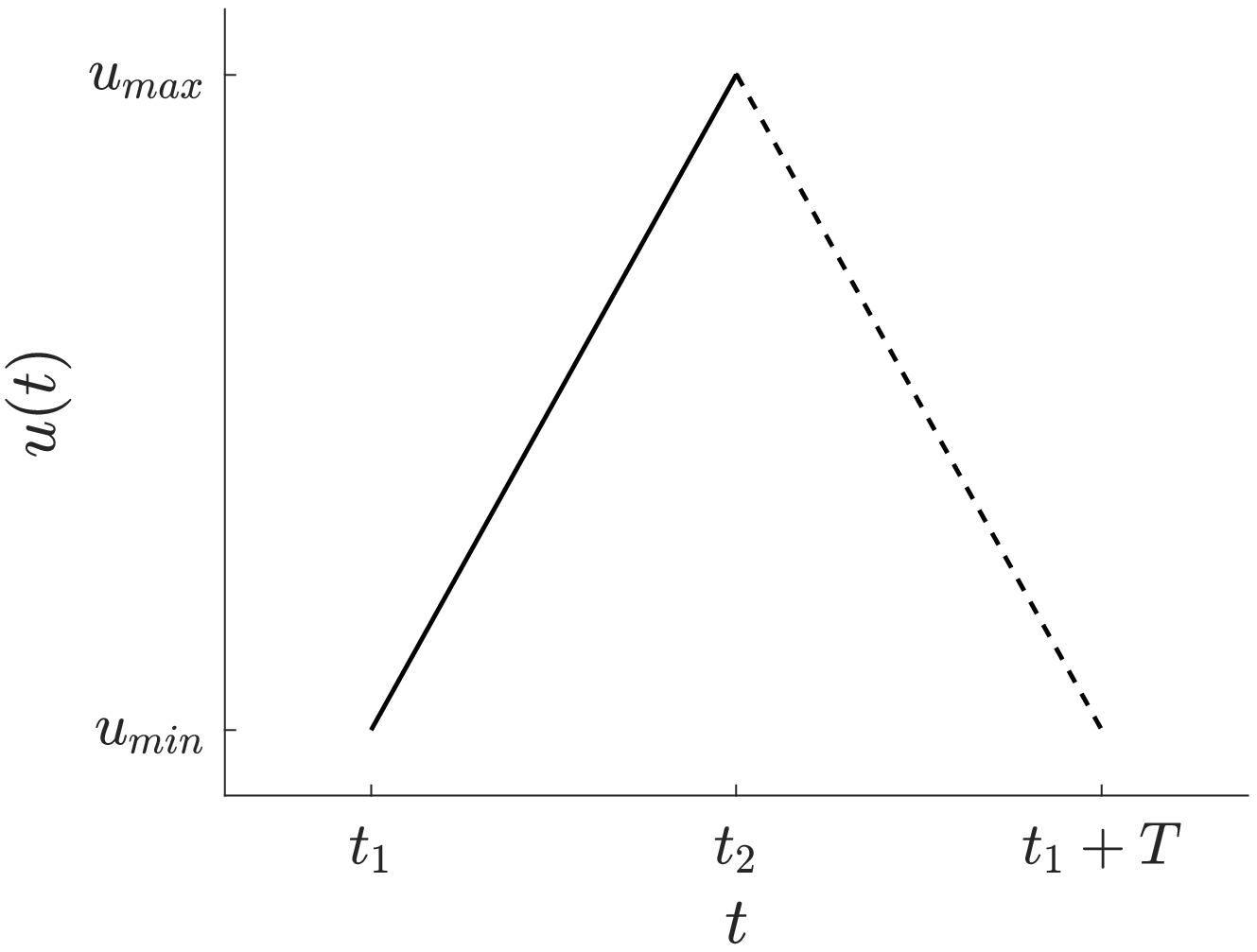}}
        \subfigure[Output signal]{\includegraphics[width=0.23\textwidth,trim={0 0 0 0.2cm},clip]{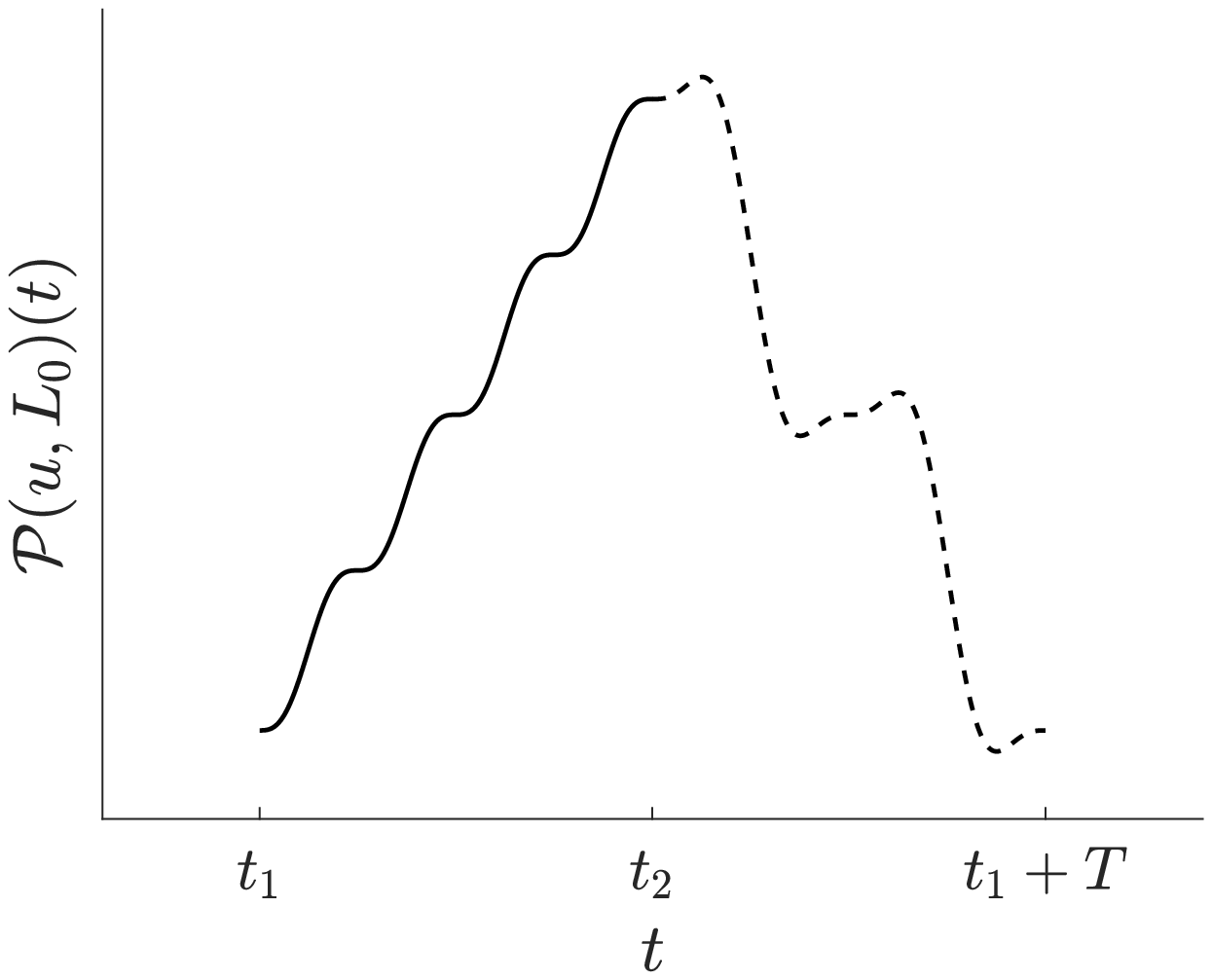}}
        \subfigure[Input-output phase plot]{\includegraphics[width=0.35\textwidth,trim={0 0 0 0.2cm},clip]{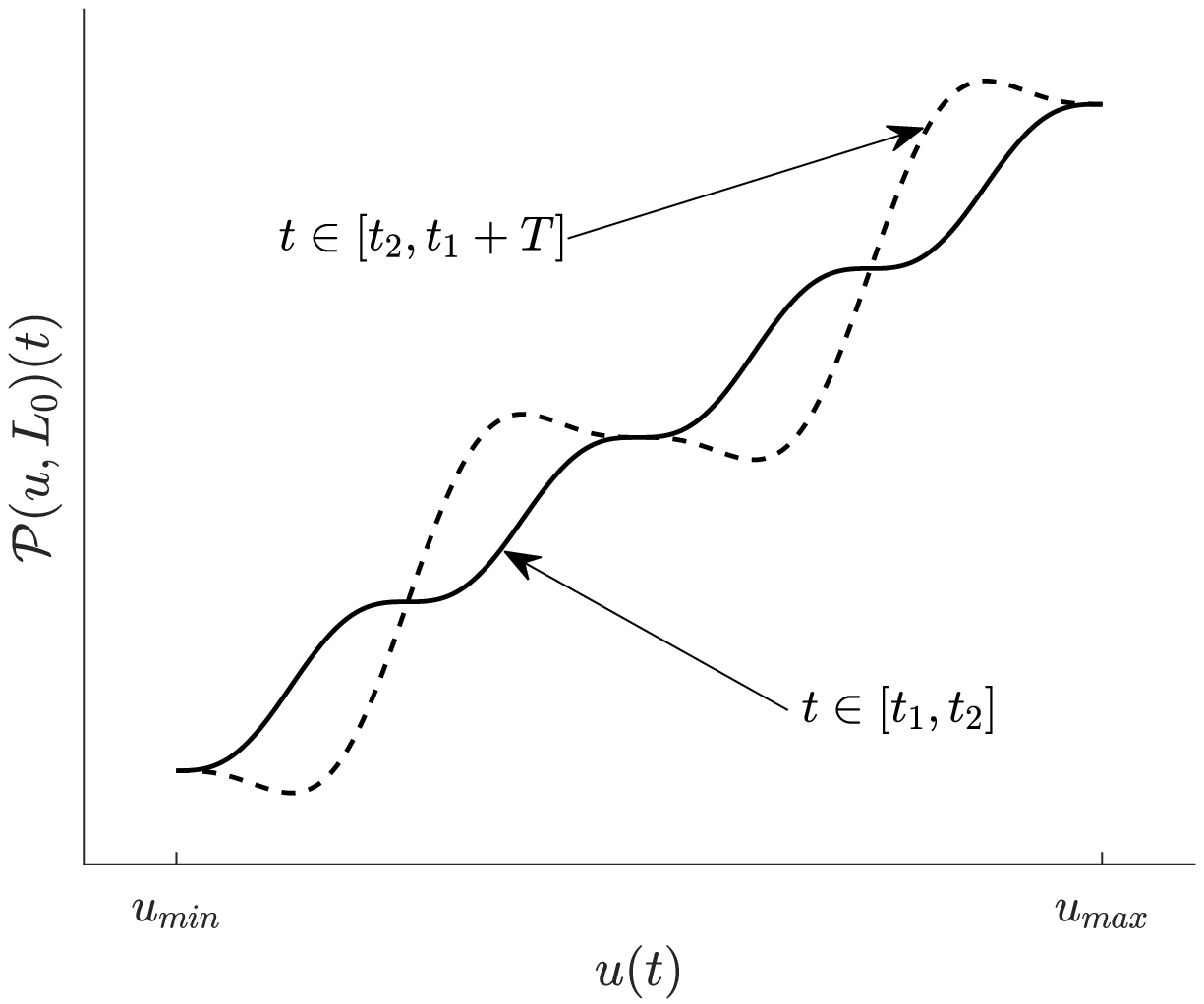}}
        \caption{Input-output hysteresis response using the multiple loops Preisach multi-loop operator. (a). The plot of input signal $u$ in a periodic time interval $[t_1,t_1+T]$. (b). The corresponding plot of output signal $y$. (c). The input-output phase plot of input and output signal which shows the hysteresis loop with more than two subloops in different orientations.}
        \label{fig:example_multiloop_phaseplot}
    \end{figure}
    
    \begin{figure}[ht]
        \centering
        \subfigure[]{\includegraphics[width=0.23\textwidth,trim={0 0 0 0.5cm},clip]{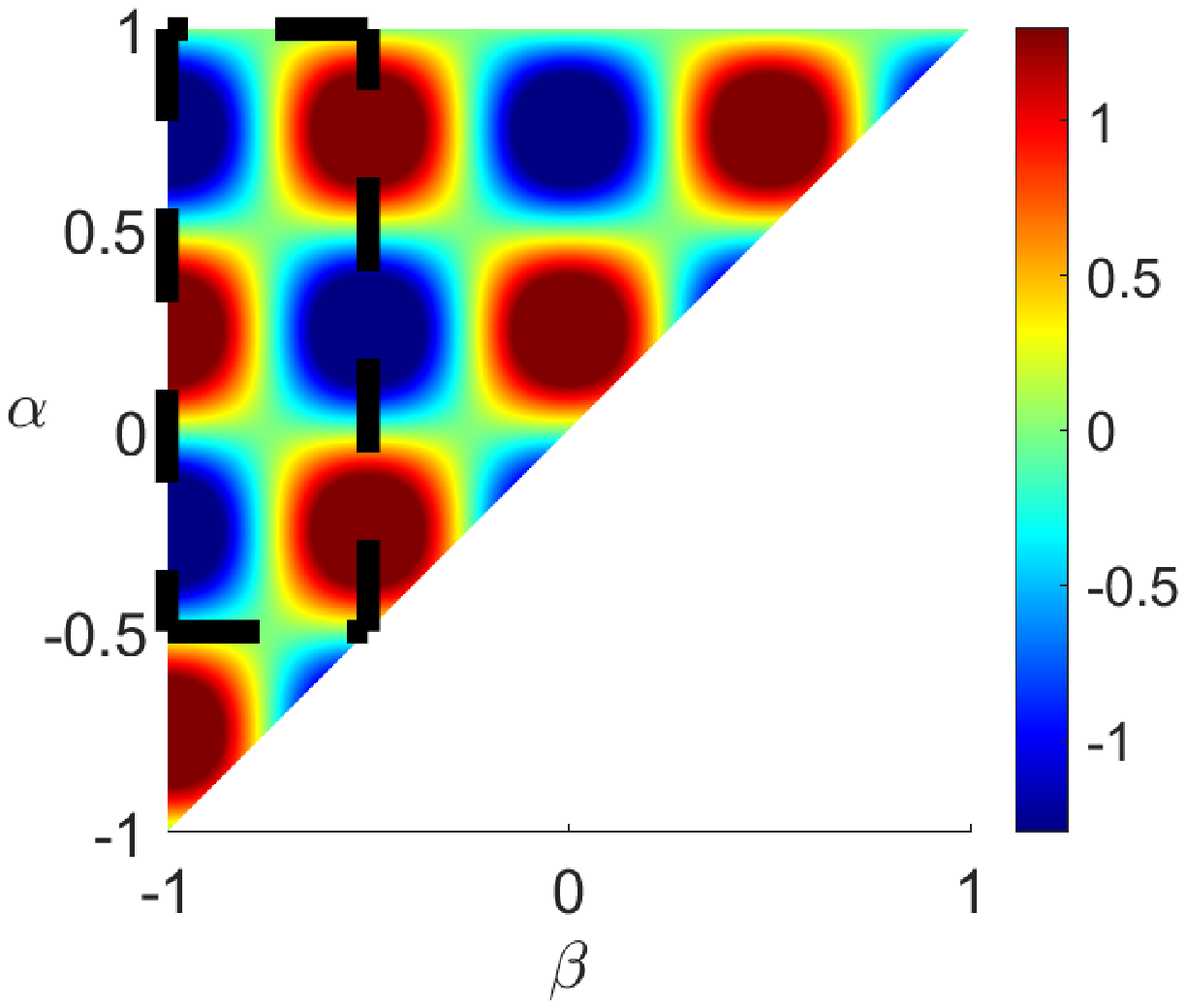}}
        \subfigure[]{\includegraphics[width=0.23\textwidth,trim={0 0 0 0.5cm},clip]{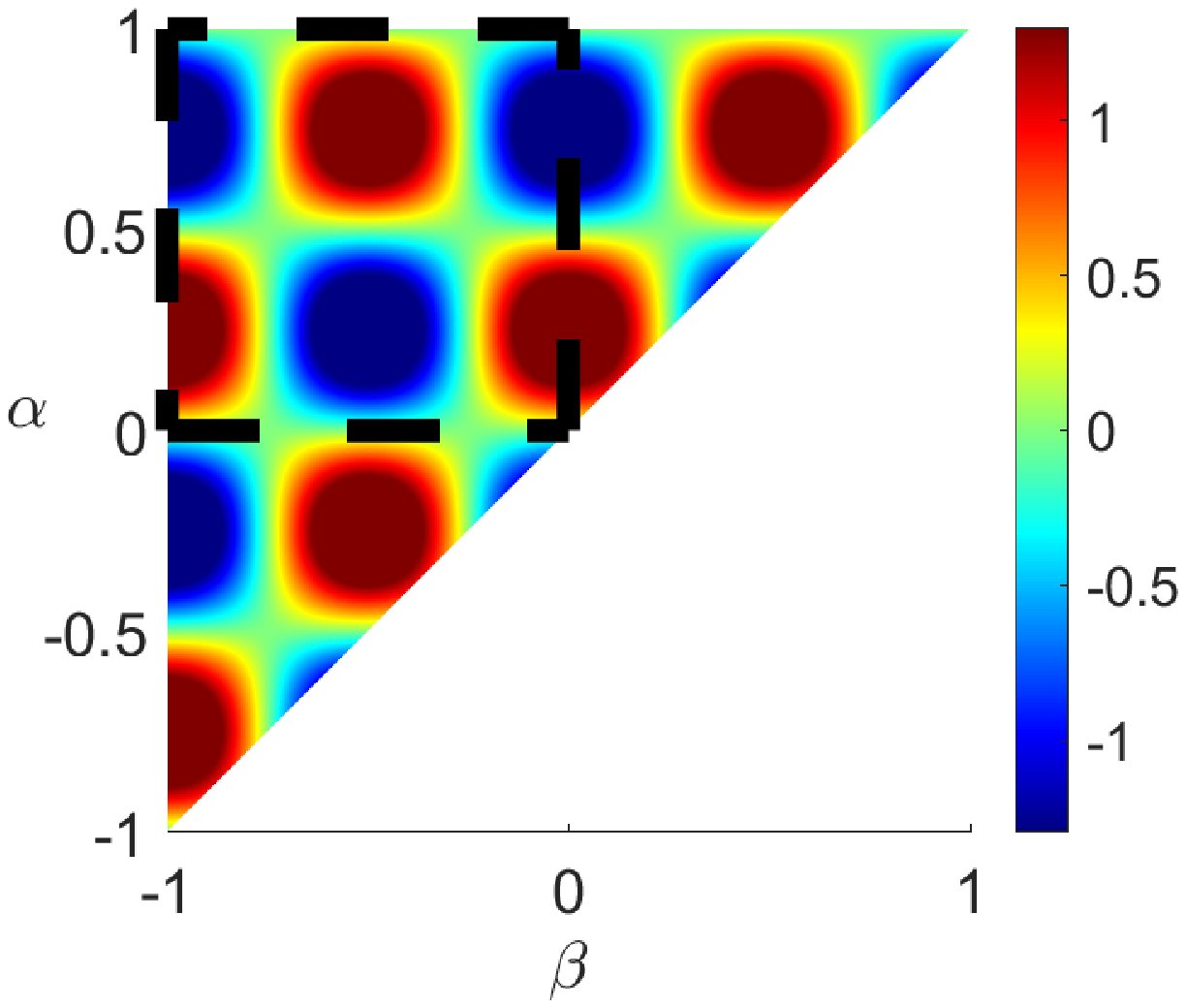}}
        \subfigure[]{\includegraphics[width=0.23\textwidth,trim={0 0 0 0.5cm},clip]{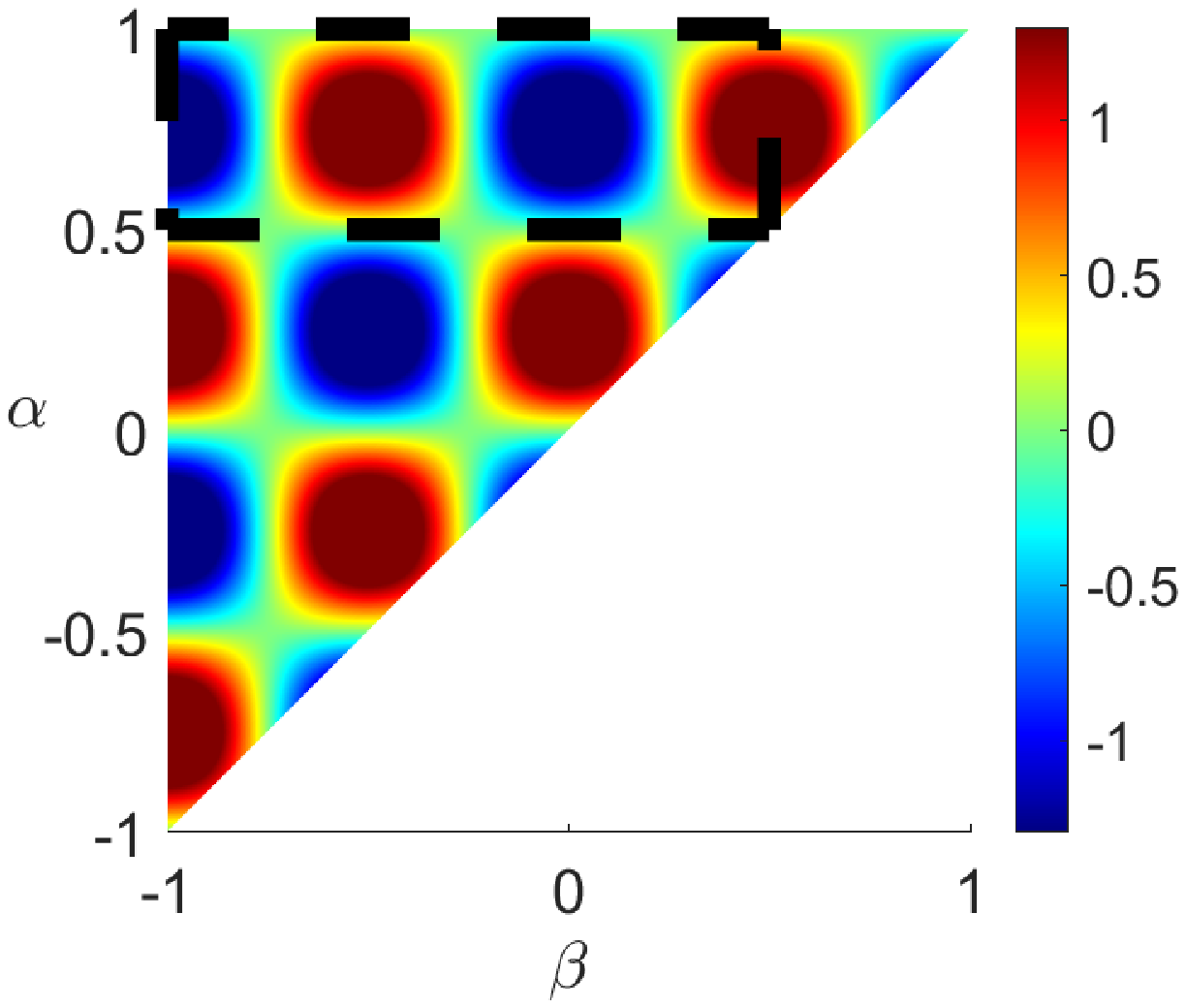}}
        \caption{Weighting function $\mu(\alpha,\beta)$ defined in \eqref{eq:example_multiloop_mu} corresponding to a {\em Preisach multi-loop operator}. The regions defined in \eqref{eq:example_multiloop_omegac} that satisfy \eqref{eq:crossover_condition} are indicated by the dashed line: $\Omega_{\alpha_1}$ in (a), $\Omega_{\alpha_2}$ in (b), and $\Omega_{\alpha_3}$ in (c). \label{fig:example_multiloop_omegac}}
    \end{figure}
\end{example}

\section{Set stability of a Lur'e system with a Preisach multi-loop operator in the feedback loop}\label{sec:lure_multiloop}
In this section, we present a brief study of the stability of a Lur'e-type system where the nonlinearity in the feedback loop is described by a Preisach multi-loop operator. We based our analysis on the results introduced in \cite{Vasquez-Beltran2020} where a bounded relation between the input and output rate of the Preisach operator has been found. Let us consider a Lur'e system which is described by 
\begin{equation} \label{eq:lure}
    \begin{aligned}    
        \Sigma_1:  &\begin{aligned} 
            \dot{x}(t) &= Ax(t) + Bw(t), \quad x(0) = x_0, \\
            z(t) &= Cx(t) \\
        \end{aligned} \\
        \Sigma_2:  &y(t) = \big( \mathcal{P}(u,L_0) \big) (t), \quad L_0\in\mathcal I, \\ 
        &\qquad\qquad\qquad\qquad\text{ with } w(t)=-y(t),\ u(t)=z(t),
    \end{aligned}
\end{equation}
where $\Sigma_1$ is a linear system, $x(t)\in\R^n$, $z(t),v(t),y(t)\in\R$ and $A,B,C$ are the system's matrices with suitable dimension, and transfer function of $\Sigma_1$ is given by $G(s) = C(sI-A)^{-1}B$. Additionally, $\Sigma_2$ is a Preisach multi-loop operator. This Lur'e system has a set of equilibria given by
\begin{equation*}
    \mathcal{E} = \{ (\bar{x},\,\bar{L}) \in \R^n \times \mathcal{I}\ |\ A\bar{x} - B\mathcal{P}( C\bar{x},\bar{L} ) = 0 \}.
\end{equation*}
Following from \cite[Proposition 3.2 \& 3.3]{Vasquez-Beltran2020}, when the weighting function $\mu$ of the Preisach operator is compactly supported, the relation between the input and output rate can be expressed by
\begin{equation*}
    \dot{y}(t) = \psi(t) \dot{u}(t), \qquad \text{a.a. } t\in \R_+,
\end{equation*}
where $\lambda_m\leq\psi(t)\leq\lambda_M$ with $\lambda_m$ and $\lambda_M$ given by
\begin{equation*}
    {\small \begin{aligned}\label{eq:psiMin}
        \lambda_m &= 2\ \min
            \left\{
            \inf_{(\gamma, \kappa) \in P}
            \displaystyle\int\displaylimits_{\kappa}^{\gamma} \mu(\gamma, \beta) \dd\beta
            ,\ 
            \inf_{(\gamma, \kappa) \in P}
            \displaystyle\int\displaylimits_{\kappa}^{\gamma} \mu(\alpha, \kappa) \dd\alpha \right\}, \\
        \lambda_M &= 2\ \max
                \left\{
                \sup_{(\gamma, \kappa) \in P}
                \displaystyle\int\displaylimits_{\kappa}^{\gamma} \mu(\gamma, \beta) \dd\beta
                ,\ 
                \sup_{(\gamma, \kappa) \in P}
                \displaystyle\int\displaylimits_{\kappa}^{\gamma} \mu(\alpha, \kappa) \dd\alpha \right\}.
    \end{aligned}}
\end{equation*}
We refer the interested readers to \cite{Vasquez-Beltran2020} for the details and proofs of these claims.
We state the next corollary which follows directly from \cite[Proposition 4.1]{Vasquez-Beltran2020}.

\begin{corollary}\label{coro:lure_stable}
    Let $\mathcal{P}$ be the Preisach multi-loop operator with a compactly supported $\mu$. Assume that $(A,C)$ is observable and $(A,B)$ is controllable. 
    Assume 
    that $\overline{G}\left(j \omega\right)$ given by
    \begin{equation*}\label{eq:Gbar}
        \overline{G}\left(j \omega\right) := \left(1 + \lambda_M G(j\omega)\right) \left(1 + \lambda_m G(j\omega)\right)^{-1},
    \end{equation*}
    is strictly positive real with $\lambda_M>0$ and $\lambda_m<0$ being the upper and lower bound of $\psi(t)$. Then $(x(t),L_t) \to \mathcal E$ as $t\to \infty$. 
\end{corollary}

\begin{example}
    Consider a Lur'e system as defined in \eqref{eq:lure} whose linear system matrices are given by
    \begin{equation*}
        \begin{aligned}
            A &= \left[\begin{matrix}
                  0 &   1 &  0 \\
                  0 &   0 &  1 \\
                -26 & -28 & -3
            \end{matrix}\right], & B &= \left[\begin{matrix} 0 \\ 0 \\ -26 \end{matrix}\right], &
            C &= \left[\begin{matrix} 1 & 0 & 0 \end{matrix}\right].
        \end{aligned}
    \end{equation*}
    Let the Preisach multi-loop operator $\mathcal{P}$ in this Lur'e system have the weighting function defined by \eqref{eq:example_multiloop_mu} in Example \ref{ex:preisach_multiloop}. It can be checked that $\lambda_M=\frac{4}{\pi}$ and $\lambda_m=-\frac{1}{2\pi}$. Moreover, it can be checked that conditions of Corollary \ref{coro:lure_stable} are satisfied. The results of a simulation of this Lur'e system with initial conditions of the linear system given by $x_0 = \left[ 0.8, -1.0, -1.0 \right]^\top$ and initial interface for the Preisach multi-loop operator given by $L_0=\{(\alpha,\beta)\in P\ |\ 0<\alpha<1,\, \beta=-0.9 \}\cup\{(\alpha,\beta)\in P\ |\ \alpha=1,\, \beta<-0.9 \}\cup\{(\alpha,\beta)\in P\ |\ \alpha=0,\, -0.9<\beta \}$ is illustrated in Fig. \ref{fig:example_stability_sim}.

    \begin{figure}[ht]
        \centering
        \subfigure[]{\includegraphics[width=0.35\textwidth,trim={0 0 0 0.5cm},clip]{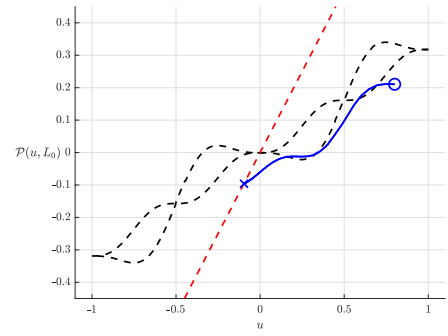}}
        \subfigure[]{\includegraphics[width=0.35\textwidth,trim={0 0 0 0.5cm},clip]{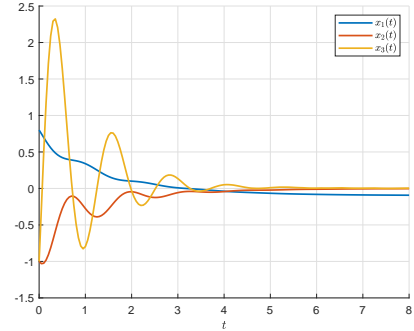}}
        \caption{Results of a simulation of a Lur'e system whose nonlinearity is the Preisach multi-loop operator with a weighting function defined as in \eqref{eq:example_multiloop_mu}. (a) Input-output phase plot. The black dashed line shows the major hysteresis loop, the red dashed line indicates the input-output pairs which correspond to states $(\bar{x},\bar{L})\in\mathcal{E}$, and the simulation is indicated by the blue line where the initial input-output $(y(0),u(0))$ is indicated by the circle and the final input-output $(y(t_f),u(t_f))$ is marked by the cross. (b) Linear system states.}
        \label{fig:example_stability_sim}
    \end{figure}
\end{example}


\section{Conclusion}\label{sec:conclusions}

In this paper we have introduced the concepts of {\em butterfly hysteresis operator} based on the characterization of the enclosed signed-area of its hysteresis loops and {\em multi-loop hysteresis operator} based on the self-intersections of its hysteresis loops. We have studied the Preisach operator and provided conditions over its weighting function such that a butterfly or a multi-loop hysteresis operator can be obtained. Moreover, we analyzed the classical problem of a Lur'e system using a Preisach multi-loop as feedback loop.



\bibliographystyle{IEEEtran}
\bibliography{bibliography}

\end{document}